\documentclass[12pt,a4paper]{amsart}
\usepackage[utf8]{inputenc}
\usepackage[english]{babel}
\usepackage{lmodern}
\usepackage{array,booktabs}
\usepackage{multicol}
\usepackage{mathrsfs, amssymb, amsmath, amsfonts}
\usepackage{graphicx}
\usepackage{geometry, xcolor}
\usepackage{enumitem}
\usepackage[abs]{overpic}
\usepackage[pagewise]{lineno} 

\title[Quasi-Linear Criticality Theory]{Quasi-Linear Criticality Theory and Green's Functions on Graphs}

\author{Florian Fischer}
\address{Florian Fischer, Institute of Mathematics, University of Potsdam, Karl-Liebknecht-Straße 24-25, 14476 Potsdam, Germany}
\email{florifis@uni-potsdam.de}

\usepackage[style=alphabetic,maxnames=5,minnames=2,abbreviate=false,date=long,url=false,isbn=false, doi=false,backend=biber, backref,backrefstyle=none]{biblatex}
\usepackage{csquotes}
\AtBeginBibliography{\tiny}
\usepackage[colorlinks=true,citecolor=blue,linkcolor=blue]{hyperref}

\newtheorem{theorem}{Theorem}[section]
\newtheorem{lemma}[theorem]{Lemma}
\newtheorem{proposition}[theorem]{Proposition}
\newtheorem{corollary}[theorem]{Corollary}

\theoremstyle{definition}

\newtheorem*{remark}{Remark} 
\newtheorem{definition}[theorem]{Definition}

\numberwithin{equation}{section}

\newcommand{\norm}[1]{\left\lVert #1 \right\rVert} 
\newcommand{\abs}[1]{\left\lvert #1\right\rvert} 
\newcommand{\set}[1]{\left\{ #1\right\} }
\newcommand{\ip}[2]{\left\langle #1, #2 \right\rangle}
\newcommand{\p}[1]{\left( #1 \right)^{\langle p-1 \rangle}}
\newcommand{\sse}{\subseteq}

\renewcommand{\epsilon}{\varepsilon}
\renewcommand{\phi}{\varphi}
\newcommand{\NN}{\mathbb{N}}

\newcommand{\RR}{\mathbb{R}}

\newcommand{\dd}{\mathrm{d}}
\DeclareMathOperator{\sgn}{sgn}
\DeclareMathOperator{\cc}{cap}
\DeclareMathOperator{\supp}{supp}
\DeclareMathOperator{\Div}{div}
\newcommand{\FF}{F}
\newcommand{\DD}{D}
\newcommand{\KK}{\tilde{K}}
\newcommand{\VV}{\mathcal{K}}
\newcommand{\MM}{\mathcal{M}}
\newcommand{\WW}{\tilde{W}}

\newcommand{\Hmm}[1]{\leavevmode{\marginpar{\tiny%
			$\hbox to 0mm{\hspace*{-0.5mm}$\leftarrow$\hss}%
			\vcenter{\vrule depth 0.1mm height 0.1mm width \the\marginparwidth}%
			\hbox to 0mm{\hss$\rightarrow$\hspace*{-0.5mm}}$\\\relax\raggedright #1}}}
			
\begin{document}

\begin{abstract}
We study energy functionals associated with quasi-linear Schrödinger operators on infinite graphs, and develop characterisations of (sub-)criticality via Green's functions, harmonic functions of minimal growth and capacities. We proof a quasi-linear version of the Agmon-Allegretto-Piepenbrink theorem, which says that the energy functional is non-negative if and only if there is a positive superharmonic function.  Furthermore, we show that a Green's function exists if and only if the energy functional is subcritical. Comparison principles and maximum principles are the main tools in the proofs. \\
	\\[4mm]
	\noindent  2020  \! {\em Mathematics  Subject  Classification.}
	Primary  \! 39A12; Secondary  31C20; 31C45; 35J62; 35R02.
	\\[4mm]
\end{abstract}

\maketitle


\section{Introduction}
One of the most fundamental questions in mathematical physics is the one of the criticality of energy functionals associated with quasi-linear Schrödinger operators. Other names for critical functionals are sometimes degenerate, parabolic or recurrent functionals. Here, a functional is called critical if the Hardy inequality holds true.

In the last two decades, criticality theory of local quasi-linear energy functionals with not necessarily non-negative potential part was studied and many characterisations of criticality in this local but non-linear setting where shown, see \cite{DAD14, DP16, HPR21, PP, PR15, PT08, PT, PTT08}. 

In the non-local quasi-linear setting, \cite{FS08} is a seminal work. In the special case of weighted infinite graphs, quasi-linear criticality theory was studied in \cite{F:GSR}. There, a ground state representation allows the connection between criticality, null-sequences and capacities.

As a first main result in this paper, we show further characterisations of criticality in terms of Green's functions, global minimal positive harmonic functions and capacity, see Theorem~\ref{thm:GreensFunction} and Theorem~\ref{thm:critical}. The ideas of constructing Green's functions and capacities are the cores in every potential theory, confer \cite{Bjoern, Prado} for metric length spaces and combinatorial graphs. Furthermore, the concept of global minimal positive  harmonic functions was firstly studied in \cite{Ag82} for the linear and local case, for linear Schrödinger operators see \cite{Pinsky95}, and for local $p$-Schrödinger operators see \cite{HPR21, PP, PR15, PT07, PT08}. We show in particular, that if a global minimal positive  harmonic function exists, it is the Agmon ground state.

As a second main result, we prove an Agmon-Allegretto-Piepenbrink theorem, see Theorem~\ref{thm:AP}. It states that the energy functional is non-negative if and only if there is a positive superharmonic function. See \cite{Allegretto, Piepenbrink, S87} for a  linear version in the continuum, \cite{Do84,KePiPo1} for a  linear version in the discrete setting, \cite{PP} for a recent non-linear version in the continuum, and \cite{LenzStollmannVeselic} for a corresponding result on strongly local Dirichlet forms. 

For the proof of the Agmon-Allegretto-Piepenbrink theorem, we show the following basic results on finite subsets of the infinite graph: a local Harnack principle, a Picone-type inequality, an Anane-D\'{i}az-Sa\'{a}-type inequality, the existence of a principal eigenvalue, the existence and uniqueness of solutions to the Poisson-Dirichlet problem, characterisations of the maximum principle. While these results are known on finite graphs (see \cite{AmghibechPicone, CL11, HoS2, HoS, ParkChung, ParkKimChung}), some adaption is needed to deal with the possibly infinite boundaries, and so we included these methods for convenience. These tools are folklore in the linear case, and they are also well understood in the quasi-linear but local case (see here \cite{HPR21, PP, PR15, PT08, PT}).

Using the toolbox above together with comparison principles will then lead to the first main result, the characterisations of criticality (see Theorem~\ref{thm:GreensFunction} and Theorem~\ref{thm:critical}).


We close this paper by showing upper and lower bounds for the principal eigenvalue, the so-called Barta's inequality, see also \cite{AllegrettoHuangPicone, AmghibechPicone,Barta, Urakawa} for versions in different settings.

In contrast, the linear case is well understood and we comment it only briefly: On Riemannian manifolds, such a theory is classical for second-order linear elliptic operators with real coefficients, see e.g. \cite{Mur86, P07, Pinsky95}. On graphs associated with linear Schrödinger operators, see \cite{KePiPo1}, and for linear Laplace-type operators, \cite[Chapter~6]{KLW21} is a rich source, and confer also the references therein. In the setting of random walks, a critical energy functional associated with a graph is usually called a recurrent graph. Here, we refer to the monographs \cite{Soa1, Woe-Book, WoessMarkov}. This linear theory is also closely related to the theory of Jacobi matrices, see here \cite{FSW08}. For more general Dirichlet forms, we refer to \cite{Fuk}, and for Schrödinger forms to \cite{Tak14, Tak16, TU21}. 


\section{Preliminaries and Main Results}

\subsection{Graphs and Schrödinger Operators}
By a \emph{(weighted) graph} $(X,m,b)$ we understand a triple of a countable infinite set $X$, strictly positive function $m\colon X\to (0,\infty)$ and a symmetric function  $b\colon X\times X \to [0,\infty)$ with zero diagonal such that $ b $ is locally summable, i.e., the vertex degree satisfies \[ \deg(x):=\sum_{y\in X}b(x,y)<\infty, \qquad  x\in X.\] The elements of $X$ are called \emph{vertices}. The function $m$ represents weights on the vertices and the function $b$ decodes weights between vertices. Two vertices $x, y$ are called \emph{connected} with respect to the graph $(X,m,b)$ if $b(x,y)>0$, in terms $x\sim y$. A set $V\sse X$ is called \emph{connected} with respect to $(X,m,b)$, if for every two vertices $x,y\in V$ there are vertices ${x_0,\ldots ,x_n \in V}$, such that $x=x_0$, $y=x_n$ and $x_{i-1}\sim x_i$ for all $i\in\set{1,\ldots, n-1}$. For $V\sse X$ let $\partial V=\set{y\in X\setminus V : y\sim z\in V}$, that is the set of all vertices outside of $V$ which are connected to $V$.
Throughout this paper we will always assume that 
\begin{center}$(X,m,b)$ is a connected graph.\end{center}

A graph $(X,m,b)$ is called \emph{locally finite} if for all $x\in X$ 
\[\# \set{y\in X : b(x,y)> 0  }< \infty.\]

Let $S$ be some arbitrary set. Then, $f\colon S\to \RR$ is called \emph{non-negative}, \emph{positive}, or \emph{strictly positive} on $I\sse S$, if $f\geq 0$, $f\gneq 0$, $f>0$ on $I$, respectively. 

The space of real valued functions on $V\subseteq X$ is denoted by $C(V)$ and the space of functions with compact support in $V$ is denoted by $ C_c(V)$. We consider $C(V)$ to be a subspace of $C(X)$ by extending the functions of $C(V)$ by zero on $X\setminus V$.  

We define the support of $f\in C(X)$ via
\[\supp (f) :=\set{ x \in X  : f(x)\neq 0}. \]
Note that 
\[\supp (f) \sse V \iff f\in C(V).\]
 We also introduce the linear difference operator $\nabla$ for all functions $f$ via
\[\nabla_{x,y}f:=f(x)-f(y), \qquad x,y\in X.\]



Now, we turn to Schrödinger operators: Let $p\geq 1$. For $V\sse X$ let the \emph{formal space} $ \FF(V)=\FF_{b,p}(V) $ be given by
\begin{align*}
\FF(V)&= \{ f\in C(X): \sum_{y\in X} b(x,y)\abs{\nabla_{x,y}f}^{p-1} < \infty  \mbox{ for all } x\in V  \}\\
&= \{ f\in C(X): \sum_{y\in X} b(x,y)\abs{f(y)}^{p-1} < \infty  \mbox{ for all } x\in V  \}.
\end{align*}
The equality above follows easily from the $p$-triangle inequality
\[\abs{\alpha+\beta}^{p}\leq 2^{p}(\abs{\alpha}^{p}+\abs{\beta}^{p}), \qquad \alpha,\beta\in \RR, p\geq 0,\]
and a proof can be found in \cite[Lemma~2.1]{F:GSR}. We set $F=F(X)$.

Note that for all $U\sse V\sse X$, we have $\FF(V)\sse \FF(U)$, and $C(U)\sse C(V)$. Moreover, for finite $K\sse X$, we have $C_c(K)=C(K)\sse \FF(K)$.

For $1< p<2$ we make the convention that $\infty\cdot 0 =0$. Then, for all $p\geq 1$ and $a\in \RR$, we set
\[ \p{a}:= |a|^{p-1} \sgn (a)=|a|^{p-2} a. \]
Here, $\sgn\colon \RR\to \set{-1,0,1}$ is the sign function, that is $\sgn(\alpha)=1$ for all $\alpha> 0$, $\sgn(\alpha)=-1$ for all $\alpha< 0$, and $\sgn(0)=0$.

Let $p\geq 1$ and $c\in C(V)$, $V\sse X$, then the quasi-linear \emph{(formal) ($p$-)Schrödinger operator} $H=H_{b,c,p,m}\colon \FF(V)\to C(V)$ is given by
\[Hf(x):=Lf(x)+\frac{c(x)}{m(x)}\p{f(x)}, \qquad x\in V,\]
where $L=L_{b,p, m} \colon \FF(V)\to C(V)$ is defined via
\[ Lf(x):=\frac{1}{m(x)}\, \sum_{y\in X} b(x,y)\p{\nabla_{x,y}f}, \qquad x\in V,\]
and is called the \emph{(formal) ($p$-)Laplacian} and the function $c$ is usually called \emph{potential} of $H$. If $c\geq 0$, then Schrödinger operators sometimes also go under the name Laplace-type operators.

Writing for all absolutely summable functions in the second variable $A$ on $X\times X$,
\[(\Div A)(x):=\frac{1}{m(x)}\sum_{y\in X}A(x,y), \qquad x\in V,\]
we immediately see the following connection to weighted $p$-Laplace-Beltrami operators in the continuum via
\[Lf(x)=\Div(b\p{\nabla f})(x), \qquad x\in V, f\in \FF(V).\]

A function $u\in \FF(V)$ is said to be a \emph{($p$-)solution, (($p$-)supersolution, ($p$-)sub-solution)} on $V\sse X$ with respect to $H$ and $g\in C(V)$ if \[Hu=g \quad(Hu\ge g,\, Hu\leq g)\qquad\text{ on }V.\] 
If $g=0$ we speak of \emph{($p$-)harmonic, (($p$-)superharmonic, ($p$-)subharmonic)} functions on $V$. If a function is superharmonic but not harmonic in $V$, we call it \emph{strictly superharmonic} in $V$. If $V=X$ we only speak of super-/sub-/harmonic functions, respectively super-/sub-/solutions with respect to $g$.

A function $u\in \FF(V)$ is said to be a \emph{($p$-)eigenfunction} to the \emph{($p$-)eigenvalue} $\lambda\in \RR$ on $V$ with respect to $H$ if \[Hu=\lambda \p{u}\qquad\text{ on }V.\]
If $u>0$ is an eigenfunction to $\lambda\in \RR$ on $V$, then $\lambda$ is called \emph{principal eigenvalue} on $V$ with respect to $H$.

\subsection{Energy Functionals Associated with Graphs}
Let $p\geq 1$ and $c\in C(X)$ arbitrary. Let $\DD=\DD_{b,c,p}$ be given by
\begin{align*}
\DD=\bigl\{f\in C(X): \sum_{x,y\in X} b(x,y)\abs{\nabla_{x,y}f}^p+ \sum_{x\in X}\abs{c(x)}\abs{f(x)}^p<\infty \bigr\}.
\end{align*}
Then, the \emph{(p-)energy functional} $h=h_{b,c,p}\colon \DD\to \RR$ is defined via 
\[ h(f)=\frac{1}{2}\sum_{x,y\in X} b(x,y)\abs{\nabla_{x,y}f}^p+ \sum_{x\in X}c(x)\abs{f(x)}^p.\]

If $p=2$ such a functional is a quadratic form, and then sometimes called Schrödinger form.

The connection between $H$ and $h$ on $C_c(X)$ can be described via a Green's formula which is stated next, and for a proof we refer to \cite{F:GSR}.

Let $V\sse X$. To shorten notation, we define a weighted bracket $\ip{\cdot}{\cdot}_{V}$ on $C(X)\times C_c(X)$ via
\[\ip{f}{\phi}_{V}=\sum_{x\in V}f(x)\phi(x)m(x),\qquad f\in C(X), \phi\in C_c(X).\]
 
\begin{lemma}[Green's formula, Lemma~2.3 in \cite{F:GSR}]\label{lem:GreensFormula}
	Let $p\geq 1$ and $V\sse X$. Let $f\in \FF(V)$ and $\phi\in C_c(X)$. Then all of the following sums converge absolutely and
	\begin{align*}
		\ip{Hf}{\phi}_{V}
		&=\frac{1}{2}\sum_{x,y\in V} b(x,y)\p{\nabla_{x,y}f}( \nabla_{x,y}\phi)+ \sum_{x\in V}c(x)\p{f(x)}\phi(x)\\
		&\qquad +\sum_{x\in V, y\in \partial V}b(x,y)\p{\nabla_{x,y}f}\phi(x). 
	\end{align*}
	In particular, the formula can be applied to $f\in C_c(X)$, or $f\in \DD$, and
	\[h(\phi)=\ip{H\phi}{\phi}_{V}, \qquad \phi\in C_c(V).\]
\end{lemma}
Moreover, another nice connection between $h$ and $H$ is that $p\cdot H$ is the G\^{a}teaux derivative of $h$ on $C_c(X)$, and since we will use this observation later several times, we include the proof here.

\begin{lemma}[G\^{a}teaux derivative]\label{lem:GateauxDerivative}
Let $p\geq 1$ and $V\sse X$. For all $\phi \in \DD$ and $\psi\in C_c(V)$ we have $\phi +\psi \in \DD$ and
	\begin{align*}
		\frac{\dd}{\dd\, t}\Bigl.  h(\phi+t\, \psi) \Bigr|_{t=0}=p\, \ip{H\phi}{\psi}_{V}.
	\end{align*}
\end{lemma}
\begin{proof} Let $\phi \in \DD$ and $\psi\in C_c(V)$. By the $p$-triangle inequality, we get $\phi +\psi \in \DD$. The formula follows now easily via Green's formula, Lemma~\ref{lem:GreensFormula},
\begin{align*}
		\frac{\dd}{\dd\, t}\Bigl.  h(\phi+t\, \psi) \Bigr|_{t=0} 
		&=\frac{p}{2}\sum_{x,y\in X}b(x,y)\p{\nabla_{x,y}\phi}\bigl(\nabla_{x,y}\psi\bigr) 
		+ p\ip{c/m}{\p{\phi}\psi}_{X}\\
		&=p \sum_{x\in X}H\phi(x)\psi(x)m(x)=p\ip{H\phi}{\psi}_{V}.\qedhere
	\end{align*}
\end{proof}
Because of Lemma~\ref{lem:GreensFormula} or Lemma~\ref{lem:GateauxDerivative}, it is convenient to define \[h(\phi, \psi):= \ip{H\phi}{\psi}_{X},\qquad \phi, \psi\in C_{c}(X),\] see e.g. \cite{StrichartzWong} for a discussion for the $p$-Laplacian and associated energy on the Sierpinski gasket. However, we will not explicitly need off-diagonal entries of the energy, and thus, stay in the following with the one-entry definition.


Assume that the functional $h$ is non-negative on $C_c(V)$, $V\sse X$. Then, $h$ is called \emph{subcritical} in $V$ if the Hardy inequality holds true, that is, there exists a positive function $w$ such that \[ h \geq \norm{\cdot}_{p,w}^{p} \qquad \text{on } C_c(V).\]
Here, \[\norm{\phi}_{p,w}:=\biggl(\sum_{x\in X}\abs{\phi(x)}^pw(x)\biggr)^{1/p}, \qquad \phi\in C_c(X).\]

	If such a $w$ does not exist, then $h$ is called \emph{critical} on $V$. Moreover, $h$ is called \emph{supercritical} on $V$ if $h$ is not non-negative on $C_c(V)$.


\subsection{Main Results}
We have three main results. The first one is a characterisation of a non-negative energy functional, and called Agmon-Allegretto-Piepenbrink-type theorem. The second result is a characterisation of a subcritical energy functional, and the third theorem is a characterisation of a critical energy functional.

Agmon-Allegretto-Piepenbrink-type theorems usually state that the non-negati-vity of the energy functional is equivalent to the existence of a strictly positive superharmonic function with respect to the corresponding Schrödinger operator.

Since in \cite{Allegretto, Piepenbrink} such  results were proven first, many versions and applications of this theorem have been established. We note \cite[Theorem~4.3]{PP} for a recent generalisation in the continuum, \cite{LenzStollmannVeselic} for a corresponding result on strongly local Dirichlet forms and \cite[Theorem~4.2]{KePiPo1} for a corresponding version for linear $(p=2)$-Schrödinger operators on graphs. We generalise the result in \cite{KePiPo1} to $p\in (1,\infty)$ and to subsets of $X$.
\begin{theorem}[Agmon-Allegretto-Piepenbrink-type Theorem]\label{thm:AP}	
	Let $p>1$ and $V\sse X$. Then the following assertions are equivalent:
	\begin{enumerate}[label=(\roman*)]
		\item\label{thm:AP1} $h$ is non-negative on $C_c(V)$;
		\item\label{thm:AP2} there exists a function which is strictly positive in $V$, vanishes in $X\setminus V$, and is superharmonic in $V$;
		\item\label{thm:AP2a} there exists a function which is strictly positive in $V$ and is superharmonic in $V$.
	\end{enumerate}
	Moreover, if the graph locally finite on the infinite set $V$, then the above is also equivalent to following assertions:
	\begin{enumerate}[label=(\roman*), start=4]
	\item\label{thm:AP3} there exists a function which is strictly positive in $V$, vanishes in $X\setminus V$, and is harmonic in $V$;
	\item\label{thm:AP3a} there exists a function which is strictly positive in $V$ and is harmonic in $V$.
	\end{enumerate}
\end{theorem}

\begin{remark}
Clearly, if the graph is finite, that is, $X$ is finite, then the non-negativity of the energy does not imply the existence of a positive harmonic function, see \cite[Corollary~0.56]{KLW21} for $p=2$ and confer also Proposition~\ref{prop:PP}. It is natural to ask whether the implication \ref{thm:AP1}~$\implies$~\ref{thm:AP3a} in Theorem~\ref{thm:AP} does also hold for infinite graphs which are not locally finite. The following example shows that this is in general not the case. This example is a slightly modified generalisation of the example for $p=2$ in \cite[p. 185]{HK} to all $p>1$.

Consider the so-called star graph: Let $(X,b,m)$ be a graph on $X=\NN_{0}$ such that for all $n,k\in \NN_{0}$ we have $b(n,k)>0$ if and only if either $n=0$ or $k=0$. Moreover, set $m=1$ and $c=1$. Thus, $h$ is non-negative. Assume that $u$ is a non-negative harmonic function. Then,
for all $n\in\NN$ we have
\[Lu(n)=b(n,0)\abs{\nabla_{n,0}u}^{p-2}\nabla_{n,0}u=-\abs{u(n)}^{p-1},\]
as well as
 \[Lu(0)=\sum_{n=1}^{\infty}b(0,n)\abs{\nabla_{0,n}u}^{p-2}\nabla_{0,n}u=-\abs{u(0)}^{p-1}.\]
 Hence, combining the two equalities we end up with
 \[\abs{u(0)}^{p-1}=-\sum_{n=1}^{\infty}\abs{u(n)}^{p-1},\]
 which implies that $u=0$.
 
 It would be very interesting to have a characterisation of graphs for which the implication \ref{thm:AP1}~$\implies$~\ref{thm:AP3a} is true. We know from \cite{F:GSR}, that having a critical energy functional in $X$ implies the existence of a unique positive harmonic function in $X$. So, further investigations are needed on non-locally finite graphs associated with subcritical energy functionals without star-like subgraphs.
\end{remark}

The second main result deals with the (non-)existence of a particular superharmonic function. This function is defined next. On Euclidean spaces, the following notion was introduced in the linear case in \cite{Ag82}, and was then extended to weighted $p$-Laplace-type and weighted $p$-Schrödinger equations in \cite{PP, PR15, PT07, PT08}. This notion is new on graphs.

\begin{definition}\label{def:minimal}
	Let $V\sse X$ be connected and  $K\sse V$ be finite. A function $u$ which is harmonic on $V\setminus K$ and positive on $V\cup \partial V$ is called \emph{positive harmonic function of minimal growth in a neighbourhood of infinity} in $V$, and is denoted by $u\in \MM(V\setminus K)$, if for any  finite and connected subset $\KK \sse V$ with $K\sse \KK$, and any positive function $v\in \FF(V\setminus \KK)$ which is superharmonic in $V\setminus \KK$, we have
	\[u\leq v \text{ on } \KK \quad\text{ implies }\quad  u\leq v \text{ in } V\setminus \KK.\]
	If $u\in \MM(V)$, then $u$ is called a \emph{global minimal positive harmonic function} in $V$. 
	
	If $u\in \MM(V\setminus \set{o})\cap F(V)$ for some $o\in V$ and $u$ is not harmonic in $o$, then $u$ is called a \emph{(global minimal positive) Green's function} in $V$ at $o$. 
	If, moreover, $Hu=1_o$ on $V$, then the Green's function $u$ at $o$ is called \emph{normalised}.
\end{definition}

\begin{theorem}[Existence of Green's Functions]\label{thm:GreensFunction}
	Let $p> 1$ and $V\sse X$ be connected. Let $h$ be non-negative on $C_c(V)$. Then, $h$ is subcritical in $V$ if and only if there exists a normalised Green's function $G_o$ in $V$ at all $o\in V$.
\end{theorem}

As a consequence of Theorem~\ref{thm:GreensFunction} we also get the following statement, which continues a characterisation given in \cite{F:GSR}. Note that by Corollary~\ref{cor:AP} (that is \cite[Proposition~5.5]{F:GSR} together with Theorem~\ref{thm:AP}), any non-negative energy functional in $V\subsetneq X$ is necessarily subcritical in $V$.

We define the \emph{variational $p$-capacity} for $V\sse X$ and $x\in V$ via \[\cc_{h} (x,V):=\inf_{\phi\in C_c(V), \phi(x)=1}h(\phi).\]

\begin{theorem}[Characterisations of Criticality]\label{thm:critical}
	Let $p>1$ and $h$ be non-negative on $C_c(X)$. Then the following statements are equivalent:
	\begin{enumerate}[label=(\roman*)]
		\item\label{thm:critical1} $h$ is critical in $X$.
		\item\label{thm:critical9} There does not exist a normalised Green's function for any $x\in X$.
	\item\label{thm:critical11} There exists a global minimal positive harmonic function in $X$. 
	\item\label{thm:critical3} $\cc_{h}(x,X)=0$ for some $x\in X$.
	\end{enumerate}
\end{theorem}
Note that in the case of $h$ being critical, the global minimal positive harmonic function is the Agmon ground state (see Section~\ref{sec:critical} for details).

If not stated otherwise, we will in the following always assume that
\[p\in (1,\infty).\]

\begin{remark}
Our main goal is to allow negative values for $c$. As a byproduct we are usually not able to use the standard toolbox from functional analysis. A workaround will very often be the method of approximating the graph by an exhaustion of $X$  with finite sets and an analysis of the corresponding limit. 

	Note that assuming $c\geq 0$ allows the usage of standard functional analytic tools which are very similar to the continuous case. For example, if $c\geq 0$, then we can use \cite[Proposition~7.6]{Showalter} and get that the subdifferential of $h$ on $W^{1,p}(X)$ is a singelton containing only the $p$-Laplace-type operator. Here, \[W^{1,p}(X)=\ell^p(X,m)\cap \DD\] is equipped with the norm \[\norm{f}_{1,p,m}= \bigl(\norm{f}^p_{p,m}+h(f)\bigr)^{1/p}.\]
\end{remark}

\section{General Principles and Inequalities}
Here, we introduce the necessary toolbox to achieve the Agmon-Allegretto-Piepenbrink theorem globally. This will be done by showing a variety of local results. The actual proof of the first main result will then follow by a limiting process. To be more specific, in this section we show
\begin{itemize}
	\item a local Harnack principle,
	\item a Picone-type inequality,
	\item an Anane-D\'{i}az-Sa\'{a}-type inequality,
	\item the existence of a principal eigenvalue on finite subsets,
	\item the existence and uniqueness of solutions to the Poisson-Dirichlet problem on finite subsets,
	\item characterisations of the maximum principle on finite subsets.
	
\end{itemize}

\subsection{A Local Harnack Principle}
The Harnack principle is a consequence of the local Harnack inequality which is stated next. Note that here, we explicitly use that $p\neq 1$. The application of the Harnack inequality and principle in the main results is also a reason for an exclusion of the case $p=1$. For Harnack inequalities in different discrete settings, see \cite{HoS2, KePiPo1, Prado}.

\begin{lemma}[Local Harnack Inequality, Lemma~4.4 in \cite{F:GSR}]\label{lem:harnackIneq}
	Let $p>1$, $K\subseteq X$ be connected and finite, and $f\in C(X)$. Then there exists a positive constant $C=C_{K,H,f}$ such that for any non-negative function $u\in \FF(K)$ such that $Hu\geq fu^{p-1}$ on $K$, we have
	\[ \max_Ku \leq C \min_K u.\]	
Furthermore, if $u(x)=0$ for some $x\in K$, then $u(x)=0$ for all $x\in K\cup \partial K$.
	
	Moreover, any non-negative function which is superharmonic and positive on a connected set $V\sse X$ is strictly positive on $V$.
\end{lemma}

Now, we want to prove the Harnack principle. We do this by dividing the statement into two partial results, Lemma~\ref{lem:compactBC} and Proposition~\ref{prop:harnack}. The technical part is extracted in the following lemma. This lemma has many analogues in other settings, see e.g. \cite{Bjoern}, and is a standard statement in $p$-potential theory. 

Let $V\sse X$ and $o\in V$ be a fixed reference point. Then, define $S^+_{o}(V)=S_o^+(V,H)$
 as follows
\begin{align*}
S^+_{o}(V):=\set{u\in \FF(V) : u(o)=1, \,\,u \text{ is superharmonic on } V, u\geq 0 \text{ on } V\cup \partial V. }
\end{align*}
\begin{lemma}[Harnack Principle]\label{lem:compactBC}
 	Let $V\sse X$ be connected, and $(V_n)$ be an increasing exhaustion of $V$ with connected subsets. Let $f_n\in C(V_n)$ such that $f_n\to f\in C(V)$ pointwise. Let $(u_n)$ be a sequence of non-negative functions such that $Hu_n= f_n u_n^{p-1}$ on $V_n$ (resp. $Hu_n\geq f_n u_n^{p-1}$ on $V_n$), and which converges pointwise on $X$ to some extended function $u$. 	
 	Then either $u(x)=\infty$ for all $x\in V$ or $u$ is non-negative on $X$ and $Hu_n\to Hu= f u^{p-1}$ pointwise on $V$ (resp. $Hu_n\to Hu\geq f u^{p-1}$ on $V$). 
 	
 	Moreover, the set $S^+_{o}(V)$ is compact with respect to the topology of pointwise convergence.
\end{lemma}
\begin{proof}
We only show the proof for the sequence of solutions. The proof for the sequence of supersolutions follows similarly. 

Let $(u_n)$ be a sequence of non-negative functions such that $Hu_n= f_nu_{n}^{p-1}$ in $V_n$. We divide the proof into several cases.

	If $\lim_{n\to\infty} u_n(x)=\infty$ for some $x\in V$. Then, by the Harnack inequality, Lemma~\ref{lem:harnackIneq}, we have $\lim_{n\to\infty} u_n(x)=\infty$ for all $x\in V$.

	If $\lim_{n\to\infty} u_n(x)=0$ for some $x\in V$. Then, again by the Harnack inequality, Lemma~\ref{lem:harnackIneq}, we have $\lim_{n\to\infty} u_n(x)=0$ for all $x\in V$. This is, however, a non-negative harmonic function on $V$.

	Now, let us assume that there exists $u\in C(X)$ such that $\lim_{n\to\infty} u_n(x)=u(x)\in (0,\infty)$ for  all $x\in V$. Then, without loss of generality, we can assume that $u_n>0$ on $V_n$. 
	
	 Moreover, assuming $Hu_n= f_nu_n^{p-1}$ on $V_n$ is equivalent to
\begin{align*}
	\sum_{y\in X, \nabla_{x,y}u_n< 0}&b(x,y)( \nabla_{y,x}u_n)^{p-1} \\
	&= \sum_{y\in X, \nabla_{x,y}u_n> 0}b(x,y)( \nabla_{x,y}u_n)^{p-1}+\bigl(c(x)-f_n(x)m(x)\bigr)u_n(x)^{p-1}
\end{align*}
for any $x\in V_n$. Furthermore, since $u_n>0$ on $V_n$
\begin{align*}
	0&\leq \sum_{y\in X, \nabla_{x,y}u_n< 0}b(x,y)\Bigl( \frac{u_n(y)}{u_n(x)}-1\Bigr)^{p-1} \\
	&=  \sum_{y\in X, \nabla_{x,y}u_n> 0}b(x,y)\Bigl( 1-\frac{u_n(y)}{u_n(x)}\Bigr)^{p-1}+c(x)-f_n(x)m(x) \\
	&\leq  \sum_{y\in X, \nabla_{x,y}u_n> 0}b(x,y)+c(x)-f_n(x)m(x) \\
	&\leq  \deg(x)+c(x)-f_n(x)m(x) \\
	&\to \deg(x)+c(x)-f(x)m(x) <\infty.  
\end{align*}
Using dominated convergence, we infer for any $x\in V_n$
\begin{align*}
	\sum_{y\in X, \nabla_{x,y}u< 0}b(x,y)&\Bigl( \frac{u(y)}{u(x)}-1\Bigr)^{p-1} \\ 
	&=  \sum_{y\in X, \nabla_{x,y}u> 0}b(x,y)\Bigl( 1-\frac{u(y)}{u(x)}\Bigr)^{p-1}+c(x)-f(x)m(x).  
\end{align*}
Multiplying both sides with $u^{p-1}(x)$ and rearranging yields in a non-negative function such that $Hu= fu^{p-1}$ on any $V_n$ for $n$ large enough. Since $(V_n)$ is an increasing exhaustion of connected sets, we get $Hu= fu^{p-1}$ on $V$.

We now turn to the statements for $S=S^+_{o}(V)$. By the pointwise convergence it follows from the previous investigations that we get $u\in S$ if even  $(u_n)_n$ is in $S$. 
	
	Furthermore, note that $V$ is connected, so for all $x\in V$ there exists a path $x_0\sim\ldots \sim x_n$ such that $x_0=o$ and $x_n=x$. Let $V=\set{x_0,\ldots, x_n}$ then we can apply Lemma~\ref{lem:harnackIneq} to some $u\in S$ and get that there exists a constant $C_x>0$ such that $C_x^{-1}\leq u(x)\leq C_x$. Hence $S$ is included in the product space $\prod_{x\in X} [C_x^{-1}, C_x]$ which is compact due to Tychonoff's theorem. By the previous part, $S$ is closed and thus, it is also compact.
\end{proof}

From Lemma~\ref{lem:compactBC} another principle will follow easily (which is sometimes also called Harnack principle). The corresponding analogue for linear Schrödinger operators on our general graphs can again be found in \cite{KePiPo1}. 

\begin{proposition}[Convergence of Solutions]\label{prop:harnack}
	Let $V\sse X$ be connected, $C>0$ and $o\in V$. Assume that we have a sequence $(u_n)_n$ in $S^+_{o}(V).$ Then there exists a subsequence $(u_{n_k})_k$ that converges pointwise to a function $u\in S^+_{o}(V)$. Furthermore, assume that either
	\begin{enumerate}[label=(\alph*)]
	\item\label{thm:harnack1} the graph is locally finite, or
	\item\label{thm:harnack2} the subsequence $(u_{n_k})_k$ is monotone increasing, or
	\item\label{thm:harnack3} there exists a $f\in \FF(V)$ such that for all $k\in\NN$ we have $u_{n_k}\leq f$ in $X$.
	\end{enumerate}
	Then $Hu_{n_k}\to Hu$ pointwise as $n_k\to\infty$ .
\end{proposition}
\begin{proof}
	By the previous lemma, Lemma~\ref{lem:compactBC}, the first statement follows easily. Now let $(u_{n_k})$ be a subsequence that converges pointwise to a function $u\in S^+_{o}(V)$.
	
	Ad \ref{thm:harnack1}: If the graph is locally finite  we have 
	\begin{align*}
	 \lim_{k\to \infty}L u_{n_k}(x)&=  \lim_{k\to \infty} \frac{1}{m(x)}\sum_{y\in X} b(x,y)\p{\nabla_{x,y}u_{n_k}}\\
	&= \frac{1}{m(x)}\sum_{y\in X} b(x,y)\lim_{k\to \infty}\p{\nabla_{x,y}u_{n_k}} =L u(x),
	\end{align*}
	since we sum over a finite number of elements. The assertion for the Schrödinger operator follows now easily.
	
	Ad \ref{thm:harnack2}: If  $(u_{n_k})_k$ is monotone increasing then we can use Lebesgue's theorem of monotone convergence to interchange summation and limit as above.
	
	Ad \ref{thm:harnack3}: If $u_{n_k}\leq f$ for all $k\in \NN$ for some $f\in F$. Then we can use Lebesgue's theorem of dominated convergence to  interchange summation and limit. 
\end{proof}

%

\subsection{Picone and Anane-D\'{i}az-Sa\'{a} Inequalities}
Here, we show an inequality which has many applications, one of them will lead the way to the desired Agmon-Allegretto-Piepenbrink theorem. %

For non-local $p$-Laplacians on graphs and on $\RR^{d}$, Picone's \emph{inequality} is a consequence of the ground state representation, see \cite[Theorem~3.1]{F:GSR}. However, this inequality is only a special case of the representation and and can also be archived more directly, see \cite[Proof of Proposition~2.2]{FS08} or \cite[Lemma~2.3]{AM}. For convenience, we show an alternative proof here. 


On finite graphs with $p$-Schrödinger operators a corresponding inequality is given in \cite[Theorem~4.1]{ParkKimChung} or see \cite[Lemma~6.2]{AmghibechPicone} for the case of the standard $p$-Laplacian. Here, we show a Picone-type inequality generalising the techniques on finite graphs to infinite graphs. 

In the continuum, there exists a so-called (pointwise) Picone \emph{identity} for the $p$-Laplacian, see e.g. \cite{AllegrettoHuangPicone, BF14, PTT08}. Both proofs of Picone's, the local and the non-local case, use a pointwise identity resp. inequality. Here, we employ the following result which also shows, why we cannot hope for an identity in the non-local case in general (without adding a remainder term). In the local case, one has an identity because of an application of the chain rule.
\begin{lemma}[Pointwise-Picone-type Inequality, Lemma~4.1 in \cite{ParkKimChung}]\label{lem:Lemma41}
	Let \\${f\colon \RR^3\to \RR}$ be such that 
	\[ f(a,b,c)= \abs{a-b}^p+ \abs{a}^p\p{c-1}+ \abs{b}^p\p{1/c-1}.\]
	Then $f\geq 0$ on $\set{(a,b,c)\in \RR^3: a,b\geq 0, c>0}$. Furthermore, $f=0$ if and only if $b=ac$.
\end{lemma}

Now we can show easily the following statement, the Picone inequality.

\begin{lemma}[Picone-type Inequality]\label{lem:Picone}
	Let $V\sse X$. Let $u,v \in C(X)$ such that $v>0$ on $V$. Then for all $x\in V$
	\begin{align*}
	\sum_{y\in V}b(x,y)\left(\abs{\nabla_{x,y}u}^{p} - \p{\nabla_{x,y}v}\bigl(\nabla_{x,y}\frac{\abs{u}^p}{v^{p-1}}\bigr)\right) \geq 0,
	\end{align*}
	where we allow the sum to be $\infty$.
	
	If $V$ is connected and $u\geq 0$, then equality in the inequality above implies $u=C\, v$ on $V$ for some constant $C>0$.  
	
	Moreover, $u=C\, v$ on $V$ for some constant $C>0$ implies that we have equality in the inequality above.
%
	
	In particular, we get for all $\phi\in C_c(V)$ and $0<v\in\FF(V)$
	\[\frac{1}{2}\sum_{x,y\in X}b(x,y)\abs{\nabla_{x,y}\phi}^{p} \geq \ip{Lv}{\frac{\abs{\phi}^{p}}{v^{p-1}}}_{V}.\]
\end{lemma}
\begin{proof}

	Firstly, we consider $u\geq 0$: Applying Lemma~\ref{lem:Lemma41} with $a=u(x), b=u(y)$ and $c=v(y)/v(x)$ yields that every summand is non-negative and hence, we get the first inequality.
	
	Moreover, if we assume equality and $u\geq 0$, then again by Lemma~\ref{lem:Lemma41} this implies that every summand vanishes. Thus, for a fixed $o\in V$ we get for all $y\in V, y\sim o$, that $u(y)=u(o)v(y)/v(o)$. Set $C=u(o)/v(o)$. Since $V$ is connected we get $u=C\, v$ on $V$.
		
	Furthermore, if $u=C\, v$ for some $C>0$ then we clearly have equality.
	
	For arbitrary $u\in C(X)$, substitute $u$ by $\abs{u}$ and use the reverse triangle inequality. 
	
	The latter statement of the lemma follows from Green's formula, Lemma~\ref{lem:GreensFormula}.
\end{proof}
\begin{remark}	Another proof of Lemma~\ref{lem:Picone} using the ideas of \cite{F:GSR, FS08} can be obtained as follows. By \cite[Lemma~2.6]{FS08} or \cite[Lemma~3.9]{F:GSR}, we have for all $p\in [1,\infty)$, $t\in [0,1]$ and $a\in \RR$ that
\[ \abs{a-t}^p \geq (1-t)^{p-1}(\abs{a}^p-t). \]
For $p>1$, this inequality is strict unless $a=1$ or $t=0$.

Using this inequality and defining $a$ and $t$ properly, we get the desired pointwise estimate to prove Lemma~\ref{lem:Picone}. Indeed, first of all note that if $u(y)=0$ for some $y\in V$, then without loss of generality, we can assume that $u(x)>0$, $x\in X$, and we get the pointwise estimate using the inequality
\[1\geq \abs{1-\alpha}^{p-2}(1-\alpha), \quad \alpha \geq 0.\]
We turn to $u(y)>0$ for all $y\in V$. Then, assume that $v(x)\geq v(y)$ (and by a symmetry argument we get the other case). Then, set $t=v(y)/v(x)$ and $a=u(x)/u(y)$ and apply the inequality. This yields in
\[\abs{\nabla_{x,y}(uv)}^{p}\geq \abs{\nabla_{x,y}v}^{p-1} \nabla_{x,y}(\abs{u}^{p}v). \]
Setting $\psi= uv$, we get the Picone inequality.
\end{remark}

The following first consequence of Picone's inequality is a discrete version of \cite[Lemma~3.3]{PP}, see also  \cite[Lemme~2]{DiazSaa} and \cite{Anane}. See also \cite[Lemma~4]{GS98} and \cite[Lemma~3.5]{PR15} for special cases of this inequality in the continuum. It is an extension of \cite[Corollary~4.1]{ParkKimChung} to infinite graphs.

The following Anane-D\'{i}az-Sa\'{a}-type inequality will be used to prove characterisations of the maximum principle on finite subsets, Proposition~\ref{prop:PP}.

\begin{proposition}[Anane-D\'{i}az-Sa\'{a}-type Inequality]	\label{prop:DiazSaaIneq}
	Let $K\sse X$ be finite. Let $u_{i}\in \FF(K)$, $u_{i}>0$ on $K$, $i=1,2$. Then, 
\begin{align}\label{eq:ADS}
\begin{aligned}
&\ip{\frac{Lu_1}{u_{1}^{p-1}}-\frac{Lu_{2}}{u_{2}^{p-1}}}{ u_{1}^{p}-u_{2}^{p}}_{K}\geq\\
	&\sum_{x\in K, y\in \partial K}b(x,y) \left(u_{1}^{p}(x)-u_{2}^{p}(x) \right)\left(\p{1-\frac{u_{1}(y)}{u_{1}(x)}}- \p{1-\frac{u_{2}(y)}{u_{2}(x)}}\right).
\end{aligned}	
\end{align}		
	Furthermore, if $K$ is connected, then equality implies $u_{1}=C u_{2}$ on $K$ for some constant $C>0$. Moreover, if $u_{1}=C u_{2}$ on $K\cup \partial K$, then we have equality, and the sums are non-negative.
	
	The right-hand side in \eqref{eq:ADS} is non-negative if for each pair $(x,y)\in K\times \partial K$ one of the following holds true:
	\begin{enumerate}[label=(\alph*)]
	\item\label{ADS1} $u_{1}(x)=u_{2}(x)$, or
	\item\label{ADS2} $u_{1}(x)u_{2}(y)=u_{1}(y)u_{2}(x)$, or
	\item\label{ADS3} $u_{1}(x)>u_{2}(x)$, and $u_{1}(x)u_{2}(y)< u_{1}(y)u_{2}(x)$, or
	\item\label{ADS4} $u_{1}(x)< u_{2}(x)$, and $u_{1}(x)u_{2}(y)> u_{1}(y)u_{2}(x)$.
	\end{enumerate}
	
	In particular, if $\phi, \psi\in C(K)$, $\phi,\psi> 0$ on $K$, i.e., $\phi=\psi=0$ on $X\setminus K$. Then, 
	\begin{align}\label{eq:ADSfinite}
		\ip{L\phi}{ \frac{\phi^{p}-\psi^p}{\phi^{p-1}}}_{K}+\ip{L\psi}{ \frac{\psi^{p}-\phi^p}{\psi^{p-1}}}_{K}\geq 0.
	\end{align}
	If $K$ is connected, then equality implies $\phi=C\, \psi$ for some $C>0$. Moreover, if $\phi=C\, \psi$ for some $C>0$, then we have equality.
\end{proposition}
\begin{proof} 
Set for all $x\in X$
\[\psi_{1}(x):=1_{K}(x)\frac{ u_{1}^{p}(x)-u_{2}^{p}(x)}{u_{1}^{p-1}(x)}, \qquad \psi_{2}(x):=1_{K}(x)\frac{ u_{2}^{p}(x)-u_{1}^{p}(x)}{u_{2}^{p-1}(x)}.\]
Note that the finiteness of $K$ implies that the following sums are all absolutely converging. Thus, we calculate using the Green's formula, Lemma~\ref{lem:GreensFormula}, 
\begin{align*}
	\sum_{x\in K}&\bigl(Lu_{1}(x)\psi_{1}(x)+Lu_{2}(x)\psi_{2}(x)  \bigl)m(x)\\
	&=\frac{1}{2}\sum_{x,y\in K}b(x,y)\left(\p{\nabla_{x,y}u_{1}}\nabla_{x,y}\psi_{1}+\p{\nabla_{x,y}u_{2}}\nabla_{x,y}\psi_{2} \right)\\
	&\qquad +\sum_{x\in K, y\in \partial K}b(x,y)\left( \p{\nabla_{x,y}u_{1}}\psi_{1}(x)+\p{\nabla_{x,y}u_{2}}\psi_{2}(x)\right) \\
	&=\frac{1}{2}\sum_{x,y\in K}b(x,y)\left(\abs{\nabla_{x,y}u_{1}}^{p}-\p{\nabla_{x,y}u_{2}}\nabla_{x,y}\frac{u_{1}^{p}}{u_{2}^{p-1}}\right) \\
	&\qquad +\frac{1}{2}\sum_{x,y\in K}b(x,y)\left(\abs{\nabla_{x,y}u_{2}}^{p}-\p{\nabla_{x,y}u_{1}}\nabla_{x,y}\frac{u_{2}^{p}}{u_{1}^{p-1}}\right) \\
&\qquad +\sum_{x\in K, y\in \partial K}b(x,y)\left( \p{\nabla_{x,y}u_{1}}\psi_{1}(x)+\p{\nabla_{x,y}u_{2}}\psi_{2}(x)\right).
\end{align*}
Using Picone's inequality, Lemma~\ref{lem:Picone}, yields in
\begin{align*}
	&\ldots \geq \sum_{x\in K, y\in \partial K}b(x,y)\left( \p{\nabla_{x,y}u_{1}}\psi_{1}(x)+\p{\nabla_{x,y}u_{2}}\psi_{2}(x)\right)\\
	&= \sum_{x\in K, y\in \partial K}b(x,y) \left(u_{1}^{p}(x)-u_{2}^{p}(x) \right)\left(\p{1-\frac{u_{1}(y)}{u_{1}(x)}}- \p{1-\frac{u_{2}(y)}{u_{2}(x)}}\right).
\end{align*}
This shows the first statement. Now we turn to the equality: We have used an inequality, which might be an equality, Picone's. Thus, we can read of Lemma~\ref{lem:Picone} that if $K$ is connected, we have $u_1=Cu_2$ on $K$ for some $C>0$.

Since $x\mapsto x^{p}, x\geq 0$, is strictly monotone increasing and $x\mapsto \p{1-x}, x\in \RR$, is strictly monotone decreasing for $p>1$, we get the desired non-negativity of the sum in the left-hand side in \eqref{eq:ADS} if \ref{ADS1}-\ref{ADS4} are fulfilled. Especially, the left-hand side is non-negative if \ref{ADS1} and \ref{ADS2} are satisfied, which is fulfilled for $u_{1}=Cu_{2}$ on $K\cup \partial K$.

Now we turn to the last assertion of the statement. Since $\phi=0=\psi$ on $\partial K$, \ref{ADS2} is fulfilled on $K\cup \partial K$. Thus, we have an equality in \eqref{eq:ADS}. This gives the desired result. 
\end{proof}
\begin{remark}
	The statement above can be generalised in the flavour of \cite[Lemma~3.3]{PP} for shifts of $u_i$, i.e., for functions $u_i +\alpha$ where $\alpha$ is a fixed constant. Since we do not need this generalisation here, we omit it.
\end{remark}

\subsection{Principal Eigenvalues on Finite Subsets}
In this subsection, we have a closer look on finite subsets of $X$. Since $X$ can be exhausted by increasing but finite subsets, the following results can be seen as a toolbox.

Let $V\sse X$, and define $\lambda_0(V)=\lambda_0(V,H)$ via
	\begin{align}\label{eq:lambda_0}
	\lambda_0(V):=\inf_{\phi\in C_c(V), \phi\neq 0}\frac{h(\phi)}{\norm{\phi}^p_{p,m}}=\inf_{\phi\in C_c(V), \norm{\phi}^p_{p,m}=1}h(\phi).
	\end{align}	

The following proposition collects and slightly generalises various results in \cite[Section~3]{ParkKimChung} and \cite[Section~4]{ParkChung}. There, only finite graphs are considered. Confer \cite[Theorem~3.9]{PP} for an analogue in the continuum. On finite graphs associated with linear Laplace-type operator such a result is also known as a Perron-Frobenius-type theorem, see \cite[Theorem~0.55]{KLW21}.
\begin{proposition}[Variational Characterisation of the Principal Eigenvalue]\label{prop:eigenvalue}
		Let \\ the set $K\sse X$ be finite. Then there exists a positive function $\phi_0\in C(K)$, i.e., $\phi_{0}=0$ on $X\setminus K$, such that
	\[ H\phi_0 = \lambda_0(K)\p{\phi_0}\qquad \text{on }K.\]
	The function $\phi_{0}$ is a minimiser of \eqref{eq:lambda_0}.
	
	Moreover, assume additionally that $K$ is connected. Then, $\lambda_0(K)$ is a principal eigenvalue on $K$, that is, there exists a strictly positive eigenfunction $\phi_0\in C(K)$ to $\lambda_{0}(K)$ on $K$.  Furthermore, $\lambda_0(K)$ is simple and any eigenvalue $\lambda>\lambda_0(K)$ of $H$ on $K$ does only have eigenfunctions which change sign.
\end{proposition}

\begin{proof}
	Let $\lambda_0=\lambda_0(K)$. The set $\{\phi\in C(K): \norm{\phi}^p_{p,m}=1\}$ is compact since $C(K)$ is finite dimensional. Hence, there exists a non-trivial minimiser $\phi_0\in C(K)$ of $h$ with $\norm{\phi_0}_{p,m}=1$. Let $t\in (-1,1)$ and consider $\phi_{t,z}=\phi_0+t\, 1_{z}$ for some fixed $z\in K$. Then $\norm{\phi_{t,z}}_{p,m}\neq 0$ and $\lambda_0\leq h(\phi_{t,z}) / \norm{\phi_{t,z}}^p_{p,m}$, i.e.,
	\[ 0\leq h(\phi_{t,z}) - \lambda_0\norm{\phi_{t,z}}^p_{p,m},\]
	where the right-hand side has a minimum in $t=0$. Hence, using Lemma~\ref{lem:GateauxDerivative},
	\begin{align*}
		0= \frac{\dd}{\dd t}\Bigl.  h(\phi_{t,z}) - \lambda_0\norm{\phi_{t,z}}^p_{p,m}\Bigr|_{t=0} = pH\phi_0(z)m(z)- \lambda_0p\p{\phi_0(z)}m(z).
	\end{align*}	
	Thus, $H\phi_0 = \lambda_0\p{\phi_0}$ on $K$.
	
	Now we show that $\phi_0$ can in fact be chosen to be non-negative on $K$. By the reverse triangle inequality, $h(\phi_0)\geq h(\abs{\phi_0})$, and $\norm{\phi_0}_{p,m}=\norm{\abs{\phi_0}}_{p,m}$. Hence, $\lambda_0\geq h(\abs{\phi_0}) /\norm{\abs{\phi_0}}_{p,m}^p$. Clearly, we have by the definition of $\lambda_0$ that $\lambda_0\leq h(\abs{\phi_0}) /\norm{\abs{\phi_0}}_{p,m}^p$. Thus, the non-negative function $\abs{\phi_0}\in C(K)$ solves the desired equation.
	
	Now, we assume for the rest of the proof that $K$ is connected. Since $\phi_0$ is non-trivial there exists $o\in K$ such that $\abs{\phi_0(o)}>0$. Thus, by the Harnack inequality, Lemma~\ref{lem:harnackIneq}, we get that $\abs{\phi_0}>0$ on $K$.
	
	We show that $\lambda_0$ is simple: We have seen that if $\phi_0$ is an eigenfunction to $\lambda_0$ on $K$, then so is $\abs{\phi_0}$. Hence, $h(\phi_0)=h(\abs{\phi_0})$, which is equivalent to 
	\[ \sum_{x,y\in K}b(x,y)\abs{\nabla_{x,y}\phi_0}^p= \sum_{x,y\in K}b(x,y)\abs{\nabla_{x,y}\abs{\phi_0}}^p.\]
	By the reverse triangle inequality we have $\abs{\nabla_{x,y}\phi_0}\geq  \abs{\nabla_{x,y}\abs{\phi_0}}$ for all $x\sim y$ in $K$. This implies that 
	\[ \abs{\nabla_{x,y}\phi_0}=  \abs{\nabla_{x,y}\abs{\phi_0}}, \qquad x\sim y \text{ in }K,\]
	which yields either $\phi_0=\abs{\phi_0}$ on $K$ or $\phi_0=-\abs{\phi_0}$ on $K$. Altogether, any eigenfunction to $\lambda_0$ has constant sign.
	
	Let $\phi_1$ be another eigenfunction to $\lambda_0$ on $K$. By the previous consideration, we can assume without loss of generality that $\phi_0, \phi_1>0$ on $K$. Then,
	\begin{align*}
	\ip{L\phi_0}{ \frac{\phi_0^{p}-\phi_1^p}{\phi_0^{p-1}}}_{K}&+\ip{L\phi_1}{ \frac{\phi_1^{p}-\phi_0^p}{\phi_1^{p-1}}}_{K}\\
	&=\ip{(\lambda_0-\frac{c}{m})\phi_0^{p-1}}{ \frac{\phi_0^{p}-\phi_1^p}{\phi_0^{p-1}}}_{K}+ \ip{(\lambda_0-\frac{c}{m})\phi_1^{p-1}}{ \frac{\phi_1^{p}-\phi_0^p}{\phi_1^{p-1}}}_{K}\\
	&=0.
%
%
	\end{align*}
	By the Anane-D\'{i}az-Sa\'{a} inequality \eqref{eq:ADSfinite}, this implies that $\phi_0=C\,\phi_1$ for some $C>0$ and hence, $\lambda_0(K)$ is simple.
	
	We still have to show that any eigenvalue $\lambda>\lambda_0(K)$ of $H$ on $C(K)$ can only have eigenfunctions which switch sign. Let $\epsilon> 0$ and assume that $0<\phi_{\lambda}\in C(K)$ is an eigenfunction to $\lambda$. Then by the Anane-D\'{i}az-Sa\'{a} inequality \eqref{eq:ADSfinite}, we have
	\begin{align*}
	0\leq \ip{L\phi_0}{ \frac{\phi_0^{p}-\epsilon^{p}\phi_\lambda^p}{\phi_0^{p-1}}}_{K}+\ip{L\phi_\lambda}{ \frac{\epsilon^{p}\phi_\lambda^{p}-\phi_0^p}{\phi_\lambda^{p-1}}}_{K}
	=\ip{\lambda_0-\lambda}{ \phi_0^{p}-\epsilon^{p}\phi_\lambda^p}_{K}.
%
%
	\end{align*}
	Choosing $\epsilon$ small enough leads to a contradiction.
\end{proof}

A consequence of the previous proposition is the following statement. For a counterpart in the continuum see \cite[Lemma~5.1]{PR15}.
\begin{corollary}\label{cor:principal}
Let $V\sse X$ be connected and $K\subsetneq V$ be finite and connected. If $\lambda_{0}(V)\geq 0$, then $\lambda_{0}(K)>0$.

In particular, if $h$ is non-negative on $C_c(X)$, then the principal eigenvalue is strictly positive on any finite and connected subset of $X$.
\end{corollary}
\begin{proof}
	Since $V\neq K$ there exists a finite and connected  subset $\KK\sse V$ such that $K\subsetneq \KK$. We clearly have $\lambda_{0}(K)\geq \lambda_{0}(\KK)\geq \lambda_{0}(V)$. So, it remains to show that the first inequality is strict. 
	
	By the previous statement, Proposition~\ref{prop:eigenvalue}, there exist strictly positive eigenfunctions $\phi\in C(K)$ and $\tilde{\phi}\in C(\KK)$ to the principal eigenvalues $\lambda_{0}(K)$ and $\lambda_{0}(\KK)$, respectively. Furthermore, we have
\begin{align*}
	(\lambda_{0}(K)&-\lambda_{0}(\KK))\norm{\phi}^{p}_{p,m}=h(\phi)-\lambda_{0}(\KK)\norm{\phi}^{p}_{p,m} \\
	&=\frac{1}{2}\sum_{x,y\in X}b(x,y)\abs{\nabla_{x,y}\phi}^{p}+ \ip{\left(\frac{c}{m}-\lambda_{0}(\KK)\right)\tilde{\phi}^{p-1}}{\frac{\phi^{p}}{\tilde{\phi}^{p-1}}}_{K}\\
	&=\frac{1}{2}\sum_{x,y\in X}b(x,y)\abs{\nabla_{x,y}\phi}^{p}- \ip{L\tilde{\phi}}{\frac{\phi^{p}}{\tilde{\phi}^{p-1}}}_{K}.
\end{align*}	
	
Assume that $\lambda_{0}(K)=\lambda_{0}(\KK)$, then the calculation above yields
\begin{align}\label{eq:cor3}
	\frac{1}{2}\sum_{x,y\in X}b(x,y)\abs{\nabla_{x,y}\phi}^{p}= \ip{L\tilde{\phi}}{\frac{\phi^{p}}{\tilde{\phi}^{p-1}}}_{K}.
\end{align}

By Green's formula, Lemma~\ref{lem:GreensFormula}, we have for the right-hand side
\begin{multline*}
	\ip{L\tilde{\phi}}{\frac{\phi^{p}}{\tilde{\phi}^{p-1}}}_{K}= \frac{1}{2}\sum_{x,y\in K}b(x,y)\p{\nabla_{x,y}\tilde{\phi}}\nabla_{x,y}\frac{\phi^{p}}{\tilde{\phi}^{p-1}}\\
	+ \sum_{x\in K, y\in \partial K}b(x,y)\p{\nabla_{x,y}\tilde{\phi}}\frac{\phi^{p}(x)}{\tilde{\phi}^{p-1}(x)}.
\end{multline*}	
For all $y\in\partial\KK$ and $x\in \KK$, we have
	\begin{align}\label{eq:cor2}
		\abs{\nabla_{x,y}\phi}^{p}-\p{\nabla_{x,y}\tilde{\phi}}\frac{\phi^{p}(x)}{\tilde{\phi}^{p-1}(x)}=\phi(x)^{p}-\tilde{\phi}(x)^{p-1}\frac{\phi^{p}(x)}{\tilde{\phi}^{p-1}(x)}= 0.
	\end{align}	
Thus, we have equality in \eqref{eq:cor3} for the summands outside of $\KK\times \KK$.

By the pointwise Picone inequality, Lemma~\ref{lem:Lemma41}, with $a=\phi(x)$, $b=\phi(y)$ and $c=\tilde{\phi}(y)/\tilde{\phi}(x)$, we infer that
	\begin{align}\label{eq:cor}
		\abs{\nabla_{x,y}\phi}^{p}-\p{\nabla_{x,y}\tilde{\phi}}\nabla_{x,y}\frac{\phi^{p}}{\tilde{\phi}^{p-1}}\geq 0,\qquad x,y\in \KK.
	\end{align}
Hence, by	\eqref{eq:cor2}, we have equality in \eqref{eq:cor3} if and only if we have equality in \eqref{eq:cor}. By Lemma~\ref{lem:Lemma41}, we have equality in  \eqref{eq:cor} if and only if $\phi= C\tilde{\phi}$ on $\KK$ for some constant $C>0$. But since $\phi= 0$ on $\KK\setminus K\neq \emptyset$, we get a contradiction to the positivity of $\tilde{\phi}$ in $\KK$. Hence, $\lambda_{0}(K)>\lambda_{0}(\KK)$.
\end{proof}

\subsection{Poisson-Dirichlet Problems on Finite Subsets}
Under the additional assumption that $h$ is positive on specific subsets, we can show  the existence of certain Poisson-Dirichlet problems. This is done next and the following lemma is a discrete analogue of \cite[Proposition~3.6 and Proposition~3.7]{PP}. The following lemma is needed in characterisations of the maximum principle on finite subsets, Proposition~\ref{prop:PP}. 
See also \cite{ParkKimChung} for finite graphs.

\begin{lemma}[Solutions of Poisson-Dirichlet Problems]\label{lem:J}
Let $K$ be finite. Let  $g\in C(K)$, and $f\in C_c(X\setminus K)$. Furthermore, define
\[K_{f}=\set{\phi\in C(X): \phi=f \text{ on } X\setminus K}\sse C_c(X).\]
Assume that $h(\phi)>0$ for all $\phi\in K_{f}$. 
Then the functional $j=j_{g}\colon \DD\to \RR$ defined via
	\[ j(\phi)= h(\phi)-p\ip{g}{\phi}_{K}, \qquad \phi\in \DD,\]
	attends a minimum in $K_{f}$. Moreover, any minimiser of $j$ on $K_{f}$ solves the Poisson-Dirichlet problem
	\begin{align*}
		\begin{cases}
			Hu&= g \quad\text{ on } K \\
			\phantom{H}u&=f \quad \text{ on } X \setminus K.		
		\end{cases}
	\end{align*}
	Furthermore, if $f \geq 0$ on $\partial K$ or $g\geq 0$ on $K$ then there exist minimiser which are non-negative on $K$. 
	
	In particular, if also $K$ is connected, then $g\gneq 0$ on $K$ and $f \geq 0$ on $\partial K$, or $g\geq 0$ on $K$, $\supp f \cap \partial K = \set{x_0}$ and $f(x_0)>0$, imply that the minimiser is unique.
\end{lemma}
\begin{proof}
%
	For all $\phi\in K_f$ with $\norm{\phi}_{p,m}=1$ we have for any $C>0$,
	\[j(C\,\phi)= C^ph(\phi)-C\,p\ip{g}{\phi}_{K}.\]
	Since $h(\phi)>0$, we have $j(C\,\phi)\to\infty$ as $C\to\infty$, i.e., $j$ is coercive. In particular, if $(\phi_n)$ is a minimising sequence of $j$, it is bounded. Thus, on the closed subset $K_f$ of the finite dimensional space $C(K\cup \supp f)$, $(\phi_n)$ has a convergent subsequence which converges to some $\phi_0\in K_f$.
	
	Moreover, since $K$ is finite, $\psi\mapsto p \ip{g}{\psi}_{K}$ is a bounded linear functional on $C(K)$ and thus, continuous. Moreover, also $h$ is lower semi-continuous. Altogether, $j$ is lower semi-continuous. Hence, $-\infty < j(\phi_0)\leq\lim_{n_{k}\to\infty}j(\phi_{n_k})=\inf_{\phi\in K_f}j(\phi).$ 

	Now, we show that any minimiser of $j$ on $K_f$ solves the Poisson-Dirichlet problem. Let $\phi$ be such a minimiser. Since for any $z\in K$ we have $\phi+t\, 1_z\in K_{f}\sse C_c(X)$ where $t\in \RR$, we calculate using Lemma~\ref{lem:GateauxDerivative},
	\begin{align*}
	0=\bigl. \frac{\dd }{\dd t}j(\phi+t\,1_z)\bigr|_{t=0}= pH\phi(z)m(z)-pg(z)m(z),
	\end{align*}
	which shows that $\phi$ solves the corresponding Poisson-Dirichlet problem.
	
	Let now assume that $f, g\geq 0$. For all $\psi\in K_{f}$ we get by the reverse triangle inequality, that $h(\psi)\geq h(\abs{\psi})$. Since $f$ is non-negative, we have $\psi=\abs{\psi}=f$ on $X\setminus K$, and $\abs{\psi}\in K_{f}$. Since $g$ is non-negative, we also get $j(\psi)\geq j(\abs{\psi})$. Thus, there exist minimiser which are non-negative on $K$. 
	
		The uniqueness follows from the Anane-D\'{i}az-Sa\'{a} inequality, Proposition~\ref{prop:DiazSaaIneq}. Here are the details: Let $u,v$ be such that $Hu=Hv=g\geq 0$ on $K$ and $u=v=f\geq 0$ on $\partial K$. By either $g\gneq 0$ on $K$ or $f(x_0)>0$ for some $x_0\in \partial K$, we get $u\neq 0\neq v$. By the Harnack inequality, we get that $u,v>0$ on $K$. Then, \eqref{eq:ADS} is fulfilled and we can apply Proposition~\ref{prop:DiazSaaIneq}. This gives
	
		\begin{align*}
	0&\leq \ip{g-\frac{c}{m}u^{p-1}}{ \frac{u^{p}-v^p}{u^{p-1}}}_{K}+\ip{g-\frac{c}{m}v^{p-1}}{ \frac{v^{p}-u^p}{v^{p-1}}}_{K}\\
	&=\ip{g(v^{p-1}-u^{p-1})}{ \frac{u^{p}-v^p}{u^{p-1}v^{p-1}}}_{K}\\
	&\leq 0.
%
	\end{align*}
	
	Again by Proposition~\ref{prop:DiazSaaIneq}, this implies that $u=C\,v$ on $K$ for some $C>0$. If $g(x)>0$ for some $x\in K$, then we get immediately from the calculation above that that $C=1$ and $u=v$ on $K$. 
	Thus, assume that $g=0$ on $K$ and $\supp f \cap \partial K = \set{x_0}$. Now $Hv=0$ on $K$, and $v= f\geq 0$ on $\partial K$, can be rewritten for all $x\in K$ as
	\begin{align*}
		\sum_{y\in K}b(x,y)\p{\nabla_{x,y}v}+c(x)v^{p-1}(x)
		&= \sum_{y\in \partial K}b(x,y)\p{v(x)-f(y)}\\ 
		&=b(x,x_0)\p{v(x)-f(x_0)}.
	\end{align*}
	On the other side, $u=Cv$, $Hu=0$ and $u= f\geq 0$ on $\partial K$, implies
	\begin{align*}
		\sum_{y\in K}b(x,y)\p{\nabla_{x,y}v} +c(x)v^{p-1}(x)
		&= \sum_{y\in \partial K}b(x,y)\p{v(x)-\frac{f(y)}{C}}\\ 
		&=b(x,x_0)\p{v(x)-\frac{f(x_0)}{C}}.
	\end{align*}	
	By the monotony of $x\mapsto \p{x}$ on $\RR$, we get $C=1$.
\end{proof}

\begin{remark}
	If we assume in the statement above that $K$ is infinite but $c\gneq 0$ on $K_f$ and $c\geq 0$ on $X$, then we can mimic the argumentation in Lemma~\ref{lem:J} on Sobolev-type spaces on graphs. Since they are reflexive Banach spaces, and $j$ is lower semi-continuous and coercive on $K_f$, we can use \cite[Theorem~1.2]{Struwe} which yields the existence of a minimiser. 
	
	In the case $p\geq 2$, the existence of a solution of the Poisson-Dirichlet problem in Lemma~\ref{lem:J} can also be proven differently: One might use the non-linear Fredholm alternative, see \cite[Theorem~12.10]{Appell}, which yields that the restriction of $H$ to any compact set is surjective. Hence, the Poisson-Dirchlet problem can be solved. 
\end{remark}

\subsection{Characterisations of the Maximum Principle on Finite Subsets}
The basic strategy to prove the Agmon-Allegretto-Piepenbrink theorem, Theorem~\ref{thm:AP}, will be to analyse increasing exhaustions of finite and connected subsets of $X$. Adding on every such subset a potential that decreases as the exhaustion increases, and analysing the corresponding limit will then yield the proof. Thus, we have to study properties of energy functionals which are not only non-negative, but strictly positive on a finite and connected  subset $K$, i.e., $\lambda_{0}(K)>0$. We show here that this is equivalent to the existence of non-negative superharmonic functions in $C(K)$. Moreover, $\lambda_{0}(K)>0$ is also equivalent to the validity of the maximum principle on $K$. 

For the following proposition confer also \cite[Theorem~4.2]{ParkKimChung} for finite graphs. Confer \cite[Theorem~5]{GS98} and \cite[Theorem~3.10]{PP} for analogue results in the continuum. Moreover, we refer to the monograph \cite{PS07} for details and history of the maximum principle  in the continuum.

Let us define what we actually mean by  the maximum principle: Let $V\sse X$. We say that $H$ satisfies 
	\begin{itemize}
		\item the \emph{weak maximum principle} on $V$ if for any function $s\in \FF(V)$ such that $Hs\geq 0$ on $V$ and $s\geq 0$ on $\partial V$ we have $s\geq 0$ on $V$, and
		\item the \emph{strong maximum principle} on $V$ if for any function $s\in \FF(V)$ such that $Hs\geq 0$ on $V$ and $s\geq 0$ on $\partial V$ we have either $s> 0$ or $s=0$ on $V$. 
	\end{itemize}		

\begin{proposition}[Characterisations of the Maximum Principle]\label{prop:PP}
	Let $K\sse X$ be finite. Consider the following assertions:
	\begin{enumerate}[label=(\roman*)]
		\item\label{prop:PP1} $H$ satisfies the weak maximum principle on $K$.
		\item\label{prop:PP2} $H$ satisfies the strong maximum principle on $K$.		
		\item\label{prop:PP0} The principal eigenvalue on $K$ is positive, i.e., $\lambda_0(K)>0$. 
		\item\label{prop:PP3} For any non-negative function $g\in C(K)$ there exists a non-negative function $u\in C(K)$ such that $Hu=g$ on $K$. This function is a minimiser of the functional $j_{g}$ defined in Lemma~\ref{lem:J} on $C(K)$. The minimiser is unique for $g=0$, and can be chosen to be strictly positive on $K$ if $g\gneq 0$.
		\item\label{prop:PP4} For any non-negative function $g\in C(K)$ there exists a unique non-negative function $u\in C(K)$ such that $Hu=g$ on $K$, which is strictly positive on $K$ if $g\neq 0$. 
	\end{enumerate}
Then, \ref{prop:PP2} $\iff$ \ref{prop:PP0} $\implies$ \ref{prop:PP3}, \ref{prop:PP2} $\implies$ \ref{prop:PP1}, and \ref{prop:PP4} $\implies$ \ref{prop:PP3}.  

Furthermore, if $K$ is connected, then all the assertions are equivalent.
\end{proposition}
%
%
\begin{proof} We set $\lambda_0= \lambda_0(K)$.

\ref{prop:PP2} $\implies$ \ref{prop:PP1}, and \ref{prop:PP4} $\implies$ \ref{prop:PP3} are trivial.

	\ref{prop:PP2} $\implies$ \ref{prop:PP0}: We prove the contraposition, therefore let us assume that $\lambda_0\leq 0$. By Proposition~\ref{prop:eigenvalue}, $\lambda_0$ is an eigenvalue of $H$ on $K$ with positive eigenfunction $\phi_0\in C(K)$. Let $\psi_0= -\phi_0$. Then,
	\[ H\psi_0 = \lambda_0\p{\psi_0}\geq 0\qquad \text{on }K,\]
	and $\psi_0 \leq 0$ on $K$. This contradicts \ref{prop:PP2}.
	
	\ref{prop:PP0} $\implies$ \ref{prop:PP2}: We prove the contraposition. Therefore let $v\in \FF(K)$ such that $Hv\geq 0$ on $K$, $v\geq 0$ on $\partial K$ and $v(x_m)\leq 0$ for some $x_m\in K$. Let $\phi\in C_c(K)=C(K)$ be defined via $\phi=v\wedge 0$ on $K$. Then, $\phi\leq 0$ and $ \phi Hv\leq 0$ on $X$. Moreover, 
	\[\sum_{x\in X}c(x)\p{v(x)}\phi(x)=\sum_{x\in K}c(x)\abs{\phi(x)}^{p}.\]
	Furthermore, for all $x\in X$ and $y\sim x$ we have
	\[\p{\nabla_{x,y}v}\phi(x)\geq \p{\nabla_{x,y}\phi}\phi(x),\]
	which implies that 
	\[\sum_{x\in X}Hv(x)\phi(x)\geq \sum_{x\in X}H\phi(x)\phi(x).\]
	Altogether, $h(\phi)\leq 0$ which gives $\lambda_0\leq 0$.
	
	 \ref{prop:PP2} \& \ref{prop:PP0} $\implies$ \ref{prop:PP3}: Firstly, the existence follows from Lemma~\ref{lem:J} which is applicable due to \ref{prop:PP0} (note that here $f=0$).
	
	Secondly, we show the uniqueness for $g=0$: Let $0\leq u, v\in C(K)$ be such that $Hu=Hv=g$ on $K$.  By \ref{prop:PP2}, we get that these solutions are either strictly positive or zero on $K$. If $u=0$ then we must have $g=0$. If there would be $v\neq 0$ such that $Hv=0$, then this would imply $\lambda_0\leq 0$ which contradicts the assumption $\lambda_0>0$, i.e., \ref{prop:PP0}. Hence, we have uniqueness for $g=0$. 
	
	Thirdly, if $g\gneq 0$, then by the discussion before and \ref{prop:PP2}, we get $u>0$ on $K$ for any $0\leq u\in C(K)$ such that $Hu=g$ on $K$.

	Now, we additionally assume that $K$ is connected.
	
	\ref{prop:PP1} $\implies$ \ref{prop:PP2}: This is a direct consequence of the Harnack inequality, Lemma~\ref{lem:harnackIneq}.
	
	\ref{prop:PP0} \& \ref{prop:PP3} $\implies$ \ref{prop:PP4}: The existence is ensured by \ref{prop:PP3} as well as the uniqueness for $g=0$. Since $K$ is connected,  we get the uniqueness for $g\gneq 0$ by Lemma~\ref{lem:J} which is applicable  because of \ref{prop:PP0}. 
	
	
	\ref{prop:PP4} $\implies$ \ref{prop:PP0}: By \ref{prop:PP4}, we can assume that $0<u\in C(K)$ and $0\lneq g\in C(K)$ be such that $Hu=g$ on $K$. Let $\phi_0$ be an eigenfunction to $\lambda_0$ on $K$. Since $K$ is connected, we have by Lemma~\ref{prop:eigenvalue} that $\phi_0>0$ on $K$. Let $C=\max_{x\in K} u(x)/\phi_0(x)$, then $u\leq C\phi_{0}$ on $K$ and by Proposition~\ref{prop:DiazSaaIneq}
	\begin{align*}
	0&\leq \ip{Lu}{\frac{u^{p}-C^p\phi_0^p}{u^{p-1}}}_{K}+ \ip{L(C\phi_0)}{\frac{C^p\phi_0^p-u^{p}}{C^{p-1}\phi_{0}^{p-1}}}_{K}\\
	&\leq \ip{\frac{c}{m}}{C^p\phi_0^p-u^{p}}_{K}+\ip{\lambda_{0}-\frac{c}{m}}{C^p\phi_0^p-u^{p}}_{K}\\
	&=\ip{\lambda_{0}}{C^p\phi_0^p-u^{p}}_{K}.
%
%
	\end{align*}
	Assume that $\lambda_0\leq 0$. This implies that we have equality in the calculation before and by Proposition~\ref{prop:DiazSaaIneq}, we get that $\phi_0= \tilde{C} u$ for some $\tilde{C}>0$. Hence, 
	\[0\geq \lambda_0 \tilde{C}^{p-1}\phi_0^{p-1}=H(\tilde{C}\phi_0)=Hu=g \geq 0 \quad \text{on }K.\]
	This contradicts the assumption $g\neq 0$ and thus $\lambda_0>0$.
%
%
\end{proof}

If $\lambda_0(K)>0$ for a connected and finite set $K\sse X$, then Proposition~\ref{prop:PP} implies the existence of local Green's functions on $K$, i.e., if $\lambda_0(K)>0$, then for any $y\in K$ there exists a function $G^K_y\colon K\to (0,\infty)$ such that
	\[ HG^K_y=1_y \qquad \text{on }K.\] 
	
One of the next tasks is to give a criterion for having a Green's function globally on the whole graph. We will later see that a global Green's function exists at all $x\in X$ if and only if $h$ is a  subcritical energy functional in $X$.

However, in the next section we do not turn to subcriticality but to non-negativity and show that  the non-negativity of energy functionals is equivalent to the existence of a globally positive superharmonic function.

\section{Global Results for Non-Negative Energy Functionals}
In this section, we show our first main result: a discrete non-linear version of the Agmon-Allegretto-Piepenbrink theorem for non-negative $p$-energy functionals associated with $p$-Schrödinger operators on $X$.

But before showing the statement, we need the following auxiliary lemma. This is \cite[Lemma~5.2 and Lemma~5.3]{F:GSR}. For convenience, we give the short proof here.

\begin{lemma}\label{lem:groundstate}
Let $V\sse X$ and $u\in\FF(V)$ be strictly positive on $V$ such that $Hu\geq g u^{p-1}\geq 0$ on $V$ for some non-negative $g\in C(V)$. Then, we have
\[h(\phi)\geq \norm{\phi}^p_{p, gm} \geq 0, \qquad \phi\in C_c(V).\]
In particular, $h$ is non-negative on $C_c(V)$.

If $V=X$ and $h$ is critical in $X$, then any strictly positive superharmonic function in $X$ is harmonic on $X$.
\end{lemma} 

\begin{proof}
The first statement is a direct consequence of Picone's inequality, Lemma~\ref{lem:Picone}. To be more specific, Picone's inequality implies
\[h(\phi)\geq \ip{u^{1-p}Hu}{\abs{\phi}^{p}}_{V}\geq \ip{g}{\abs{\phi}^{p}}_{V}\geq 0.\]
%

The second statement is now a direct consequence. Indeed, let $u$ be a strictly positive superharmonic function in $V=X$. Set $w=(Hu)/u^{p-1}\geq 0.$ Then we can use the previous calculation with $g=w$ to derive $h\geq \norm{\cdot }^p_{p, w m}$ on $C_c(V).$ Because $h$ is critical we get $w=0$, and $u$ is harmonic in $V$.
\end{proof}

\begin{remark}
The first statement in Lemma~\ref{lem:groundstate} can be sharpened and generalised via the ground state representation in \cite[Theorem~3.1 and Corollary~3.2]{F:GSR}. This representation basically says that the energy functionals associated with a Schrödinger operator can be bounded from below and above by so-called simplified energy functionals which consist of non-negative terms only, and a norm-term as in the statement above. Then, the inequality of Lemma~\ref{lem:groundstate} readily follows. 

Note that if $V$ is connected, we get by the Harnack inequality that any positive superharmonic function is strictly positive, and thus the assumption in Lemma~\ref{lem:groundstate} can then be weakened softly.

The second statement in Lemma~\ref{lem:groundstate} does not make sense on proper subsets of $X$, since $h$ would then necessarily be subcritical, confer \cite[Proposition~5.5]{F:GSR} and also Corollary~\ref{cor:AP}.
\end{remark}


\begin{proof}[Proof of Theorem~\ref{thm:AP}]
	\ref{thm:AP1}~$\implies$~\ref{thm:AP2}: Firstly, assume that $V$ is connected. Let $(K_n)$ be an increasing exhaustion of $V\sse X$ with finite and connected sets, and let $o\in K_1$. Moreover, let $H_n$ be the Schrödinger operator we obtain by adding $m/n$ to the potential $c$ of $H$, $n\in \NN$. Then by the definition of $\lambda_0$ in  \eqref{eq:lambda_0}, for all $n\in \NN$
	\[\lambda_0(K_n,H_n)\geq 1/n>0.\]
	Hence, by Proposition~\ref{prop:PP}, for any sequence $(g_n)$ of positive functions on $K_n$ there exists a unique positive function $u_n\in C(K_n)$ such that $H_nu_n=g_n$ on $K_n$. 
	
	Fix $n_0\in\NN$. Then, for all $n\geq n_0$
	\[ H_{n_0}u_n=\bigl(1/n_0-1/n\bigr)u^{p-1}_n+g_n\geq 0\quad\text{on }K_{n_0}.\]
	Hence, $(u_n)_{n\geq n_0}$ is a sequence of superharmonic functions on $K_{n_0}$ with respect to $H_{n_0}$. Without loss of generality, we choose $g_n$ such that $u_n(o)=1$ for all $n\in\NN$ (take e.g. $g_n=C_n\cdot 1_{x_n}$, $x_n\in K_n$, and specify the positive constant $C_n$ accordingly). Thus, we have that $(u_n)_{n\geq n_0}$ is in $S^{+}_{o}(K_{n_0},H_{n_0})$. Applying the convergence of solutions principle, Proposition~\ref{prop:harnack}, we get the existence of a pointwise converging subsequence $(u_{n_i})$ to some $u\in S^{+}_{o}(K_{n_0},H_{n_0})$ for all $n_0\in \NN$.  In particular, $u$ is superharmonic on $K_{n_0}$ with respect to $H_{n_0}$ and by the Harnack inequality, Lemma~\ref{lem:harnackIneq}, we have that $u>0$ on $K_{n_0}$. 

	Furthermore, we notice that for all $x\in K_{n_0}$,
	\begin{align*}
		0\leq H_{n_0}u(x) = Hu(x)-\frac{1}{n_0}u(x),
	\end{align*}
	i.e., $u/n_0\leq Hu$ on $K_{n_0}$, and $u$ is superharmonic on $K_{n_0}$. Letting $n_{0}\to \infty$, we get the desired positive superharmonic function on $V$ with respect $H$. Since $V$ is connected, we get by the Harnack inequality, Lemma~\ref{lem:harnackIneq}, that this positive superharmonic function is strictly positive in $V$. By the construction, we get that $u=0$ on $X\setminus V$, and the first implication is proven for connected $V$.
	
	Secondly, if $V$ is not connected, then we can decompose it in a (possibly finite) sequence of connected components $(V_i)$. By the construction above, we get in every $V_i$, that there is a function $u_i\in \FF(V_i)$ such that $u_i$ is strictly positive and superharmonic on $V_i$, and vanishes on $X\setminus V_i$. Hence, we can simply add all $u_i$, to obtain the desired function. To be more precise, the pointwise defined function $u=\sum_{i}u_i$ is strictly positive and superharmonic on $V$, and vanishes on $X\setminus V$.
	
	\ref{thm:AP2}~$\implies$~\ref{thm:AP2a}: This is trivial.
	
	\ref{thm:AP2a}~$\implies$~\ref{thm:AP1}: This is ensured a consequence of Picone's inequality, Lemma~\ref{lem:groundstate}.
%
%

Now, we assume that the graph is locally finite on $V$. 

\ref{thm:AP3}~$\implies$~\ref{thm:AP3a}: This is trivial.

\ref{thm:AP3a}~$\implies$~\ref{thm:AP1}: This follows also from Lemma~\ref{lem:groundstate}.

\ref{thm:AP1}~$\implies$~\ref{thm:AP3}: Here we follow the proof of \ref{thm:AP1}~$\implies$~\ref{thm:AP2} verbatim, where we note that $Hu=0$ follows then from the convergence of solutions principle, Proposition~\ref{prop:harnack}.
\end{proof}

Comparing \cite[Proposition~5.6]{F:GSR} and Theorem~\ref{thm:AP}, we obtain the following result.
\begin{corollary}\label{cor:AP}
	Let $V\subsetneq X$. Then $h$ is non-negative in $C_c(V)$ if and only if $h$ is subcritical in $V$.
\end{corollary}
\begin{proof}
	Clearly, if $h$ is subcritical in $V$, then $h$ is also non-negative on $C_c(V)$. On the other side, if $h$ is non-negative in $C_c(V)$, then by the Agmon-Allegretto-Piepenbrink theorem, Theorem~\ref{thm:AP}, this implies the existence of a positive  superharmonic function on $V$. Thus, we can use \cite[Proposition~5.6]{F:GSR}, and get that $h$ is subcritical in $V$.
\end{proof}

\begin{remark}
	Corollary~\ref{cor:AP} has the following interpretation for $p=2$, $m=\deg$, and $c= 0$: Given any connected graph, the induced graph on any proper subset is then a graph with boundary and thus transient.
\end{remark}

\section{Comparison Principles}
With the proof of Theorem~\ref{thm:AP}, we have shown the first main result. We  have seen that the existence of a positive superharmonic function is equivalent to the non-negativity of the energy functional. Now, we want to show that the existence of a specific positive function -- the so-called Green's function -- is equivalent to the subcriticality of the energy functional. Comparison principles will be the toolbox for the proof of the equivalence.

In the linear case, comparison principles also go under the name minimal principles, see e.g. \cite{KePiPo1}.

In this section, we show the discrete counterpart to \cite[Section~5]{PP}. Firstly, we show the comparison principle for  $p$-Schrödinger operators with non-negative potentials, see \cite[Theorem~2.3.2]{Prado} for the $p$-Laplacian on locally finite graphs and \cite{GS98, PP} for the continuous case. Secondly, we will use a sub/supersolution technique to allow negative values of the potential terms. Thirdly, we will prove a weak comparison principle for arbitrary Schrödinger operators on finite subsets which can be seen as a discrete version of \cite[Theorem~5.3]{PP}.
\subsection{A Weak Comparison Principle for Non-Negative Potentials}
The following lemma is the non-linear version of \cite[Theorem~1.7]{KLW21}, see also \cite[Theorem~3.14]{HoS2} for the standard $p$-Laplacian on finite graphs, and \cite[Theorem~2.3.2]{Prado}.

\begin{lemma}[Weak Comparison Principle for $c\geq 0$]\label{lem:WCPNonNegativePotential}
	Let $V\subsetneq X$ and $c\geq 0$ on $V$. Furthermore, let $u, v\in \FF(V)$ such that 
	\begin{align*}
	\begin{cases}
	Hu&\leq Hv \quad \text{on } \phantom{\partial}V, \\	
	\phantom{H}u&\leq \phantom{H}v \quad \text{on }\, \partial V.	
	\end{cases}
	\end{align*}
Assume that $(v-u)\wedge 0$ attains a minimum in $V$. 
	Then, $u\leq v$ on $V$.
	
	Moreover, in each connected component of $V$ we have either $u=v$, or $u<v$.
\end{lemma}
\begin{proof} Without loss of generality, we can assume that $V$ is connected. Otherwise, we do the following proof in every connected component of $V$.
 
Assume that there exists a $x\in V$ such that $v(x)\leq u(x)$. Since  $(v-u)\wedge 0$ attains a minimum in $V$, there exists $x_0\in V$ such that $v(x_0)-u(x_0)\leq 0$, and $v(x_0)-u(x_0)\leq v(y)-u(y)$ for all $y\in V$. Since $u\leq v$ on $\partial V$, we get $\nabla_{x_0,y} v\leq \nabla_{x_0,y}u$ for all $y\in V\cup \partial V$. Furthermore, we have
\begin{align*}
	0&\leq m(x_0)\bigl(Hv(x_0)-Hu(x_0) \bigr) \\
	&= \sum_{y\in V\cup \partial V}b(x_0,y)\left( \p{\nabla_{x_0,y}v} -\p{\nabla_{x_0,y}u} \right)\\
	&\qquad + c(x_0)\left( \p{v(x_0)} -\p{u(x_0)} \right)\\
	&\leq 0,
\end{align*}
where the second inequality follows from the monotony of $\p{\cdot}$ on $\RR$, and $c(x_0)\geq 0$. Thus, we have in fact equality above. 

Since $c(x_0)\geq 0$, then we get that $u(y)-v(y)$ is a non-negative constant for all $y\sim x$. By iterating this argument and using that $V$ is connected, we get that $u(y)-v(y)$ is a non-negative constant for all $y\in V\cup \partial V$. Since $u\leq v$ on $\partial V$, we conclude that $u=v$ on $V$.
\end{proof}

\subsection{A Weak Comparison Principle on Finite Subsets}\label{sec:weak}
The goal of this subsection is to derive a similar statement as in the previous lemma  for arbitrary potentials (see Proposition~\ref{prop:WCPh>0}). 

The following lemma is the discrete analogue of \cite[Proposition~5.2]{PP}. The strategy of its proof is to use Lemma~\ref{lem:WCPNonNegativePotential} for the absolute value of the potential part. 

\begin{lemma}[Sandwiching lemma]\label{lem:WCPEinkesselungh>0}
	Let $K\sse X$ be finite. Let $0\leq g\in C(K)$, $0\leq f\in C_c(X\setminus K)$ such that $\lambda_{0}(\supp f \cup K) >0$. Moreover, let $u, v\in \FF(K)$, such that 
	\begin{align*}
	\begin{cases}
		Hu\leq g\leq Hv \quad \text{on } K, \\
		\phantom{H}u\leq f\leq v\phantom{H} \quad \text{on } \partial K \cup \supp f,\\
		\phantom{H}0\leq u\leq v\phantom{H} \quad \text{on } K.
	\end{cases}
	\end{align*}	 
	Then there exists $0\leq w\in C(K \cup \supp f)$ such that 
	\begin{align*}
	\begin{cases}
		Hw= g &\quad \text{on } K, \\
		\phantom{H}w=f &\quad \text{on } \partial K \cup \supp f,\\
		u\leq w\leq v  &\quad \text{on } K.
	\end{cases}
	\end{align*}	
Moreover, assume that $K$ is connected. Then, $g\gneq 0$ on $K$, or $\supp f \cap \partial K = \set{x_0}$ and $f(x_0)>0$, implies that $w$ is unique.
\end{lemma}
\begin{proof}
	Let $\KK:=K\cup \supp f$, and
	
	\[ \VV=\set{w\in C(\KK): 0\leq u\leq w\leq v \text{ in } K},\]
	and consider $G\colon K\times \VV\to \RR$ defined via
	\[ G(x,w):=g(x)+2\cdot \frac{c_-(x)}{m(x)}\cdot w^{p-1}(x)\geq 0, \qquad x\in K, w\in \VV,\]
	where $c_-(x)
	=0\vee(-c(x)), x\in X$. Since $\lambda_{0}(\KK)>0$, we can use Lemma~\ref{lem:J}, and get the existence of $\tilde{w}\in C(\KK)$ such that
	\begin{equation}
	\begin{aligned}\label{eq:WCP21}
	\begin{cases}
		H_{\abs{c}}\tilde{w}= G(\cdot,w) &\quad \text{in } K, \\
		\phantom{H_{\abs{c}}}\tilde{w}= f &\quad \text{on } X\setminus K,
	\end{cases}
	\end{aligned}
	\end{equation}		
	where $H_{\abs{c}}:=H_{b,\abs{c}, p, m}$. Let $T\colon \VV\to \DD$, $Tw=\tilde{w}$. Then $T$ is monotone. Indeed, let $w_1, w_2 \in \VV$, $w_1\leq w_2$, then for $x\in K$, 
	\[H_{\abs{c}}(Tw_1(x))=G(x,w_1)\leq G(x,w_2)=H_{\abs{c}}(Tw_2(x)),\] 
	and $Tw_1=f=Tw_2$ on $X\setminus K$. By Lemma~\ref{lem:WCPNonNegativePotential}, we get $Tw_1\leq Tw_2$ on $K$.
	
	Moreover, if $w\in \FF(K)$ is a subsolution of
	\begin{equation}
	\begin{aligned}\label{eq:WCP22}
	\begin{cases}
		H\hat{w}= g \quad \text{in } K, \\
		\phantom{H}\hat{w}= f \quad \text{on } X\setminus K,
	\end{cases}
	\end{aligned}
	\end{equation}		
	then, $H_{\abs{c}}w(x)=Hw(x)+G(x,w)-g(x)\leq G(x,w)$ for $x\in K$ and hence, $w$ is a subsolution of \eqref{eq:WCP21}. Furthermore,  $Tw$ is a solution of \eqref{eq:WCP21} and by Lemma~\ref{lem:WCPNonNegativePotential} we get $w\leq Tw$ on $K$. Hence,
	\[ HTw=g+2\cdot\frac{c_-}{m}\bigl(\p{w}- \p{Tw}\bigr)\leq g, \quad \text{on } K,\]
	and $Tw$ is a subsolution of \eqref{eq:WCP22}.
	
	Analogously, we get that if $w$ is a supersolution of \eqref{eq:WCP22}, then $Tw$ is a supersolution of \eqref{eq:WCP22} and $Tw\leq w$ on $K$.
	
	Define the sequences $(u_n), (v_n)$ as follows: $u_1=u$, $u_n=T(u_{n-1})=T^nu$ and $v_1=v$, $v_n=T(v_{n-1})=T^nv$. Then $u\leq u_n\leq v_n\leq v$ for all $n\in \NN$, i.e., both sequences are monotone and bounded, and thus they converge pointwise monotonously on $X$, say to $u_{\infty}$ and $v_{\infty}$, respectively. Using the Harnack principle, Proposition~\ref{prop:harnack}, for monotone and dominated convergence we infer
	\[
		Hu_{\infty}=\lim_{n\to\infty}Hu_n=g+2\frac{c_-}{m}\lim_{n\to\infty}\bigl(\p{u_{n-1}}-\p{u_n}\bigr)=g \quad \text{on } K,
	\]
	and analogously, $Hv_{\infty}=g$ on $K$. Thus, $u_\infty$ and $v_\infty$ are candidates for $w$. 
	
	The uniqueness follows now from Lemma~\ref{lem:J}.
\end{proof}
Note that in the case of $c\geq 0$, the local statement of Lemma~\ref{lem:WCPEinkesselungh>0} can be obtained globally using properties of Sobolev-type spaces, i.e., of reflexive Banach spaces. 

The following proposition is a discrete analogue of \cite[Theorem~5.3]{PP}, confer also \cite{GS98}. Recall that by Corollary~\ref{cor:principal}, we get from $h\geq 0$ on $C_c(X)$ that $\lambda_{0}(\KK)>0$ for every \emph{connected} and finite subset of $X$. We highlight also that the proof of the proposition needs that the Harnack inequality can be applied to the smaller set, and thus, we first have to consider connected components of our finite set $K$.

\begin{proposition}[Weak Comparison Principle for Finite Subsets]\label{prop:WCPh>0}
Let $K\sse \KK\sse X$, where $K$ and $\KK$ are finite, and $\lambda_{0}(\KK)>0$. Moreover, let $v \in \FF(K)$ be such that $Hv\geq 0$ on $K$ and $v\geq 0$ on $\partial K \cup \KK\setminus K$. Let $u\in \FF(K)$ such that
	\begin{align*}
		\begin{cases}
			Hu\leq Hv \quad \text{on } K, \\
			\phantom{H}u\leq v\phantom{H}\quad \text{on }\partial K \cup \KK\setminus K.
		\end{cases}
	\end{align*}
If either
\begin{enumerate}[label=(\alph*)]
\item\label{prop:WCPh>0v} $v\in C(\KK)=C_c(\KK)$, i.e., $\supp (v)\in \KK$, or
\item\label{prop:WCPh>0u} $u\in C(\KK)$, $Hu\geq 0$ on $K$, and  $u\geq 0$ on $\partial K \cup \KK\setminus K$,
\end{enumerate}
	then $u\leq v$ on $K$.
\end{proposition}
\begin{proof}
	If $c\geq 0$ on $K$ and arbitrary on $X\setminus K$ then the statement follows from Lemma~\ref{lem:WCPNonNegativePotential}. Thus, we can assume without loss of generality, that $c\neq 0$ on $K$.
	
	
	Assume, initially, that $K$ is also connected.
	
	Note that $\lambda_{0}(K)\geq\lambda_{0}(\KK)>0$. By the strong maximum principle, Proposition~\ref{prop:PP} (\ref{prop:PP0} $\implies$ \ref{prop:PP2}), we conclude that either $v=0$ or $v>0$ on $K$. If $v=0$ on $K$, then by the connectedness of $K$, we can apply the Harnack inequality, Lemma~\ref{lem:harnackIneq}, and get $v=0$ on $K\cup \partial K$, and thus, $Hv=0$ on $K$. Hence, $Hu\leq 0$ on $K$ and $u\leq 0$ on $\partial K$. Applying the weak maximum principle to $-u$, we get that $u\leq 0$ on $K$.
	
	Now assume that $v>0$ on $K$ and define $C=1\vee (\max_K u/\min_K v)$, then using the assumptions on $u$ and $v$, we see that $u\leq C\, v$ and $C^{-1}u\leq v$ in $K\cup \partial K$. Moreover, by  Proposition~\ref{prop:PP}, we can assume that $Hv \gneq 0$.
	
	Firstly, assume that \ref{prop:WCPh>0v} holds. Furthermore, let $0\lneq g:=Hv$ and $f:=v$ on $X$,  and consider for a function $\tilde{v}\in \FF(K)$ the problem 
		\begin{equation}
	\begin{aligned}\label{eq:WCP31}
	\begin{cases}
		H\tilde{v}= g \quad \text{in } K, \\
		\phantom{H}\tilde{v}= f \quad \text{on } X\setminus K.
	\end{cases}
	\end{aligned}
	\end{equation}	
	Then  $C\, v$ is a supersolution of \eqref{eq:WCP31}. By Lemma~\ref{lem:WCPEinkesselungh>0}, there exists a unique solution $w\in C(\KK)$ of \eqref{eq:WCP31} such that $u\leq w\leq C\, v$ on $K$ and $w=v$ on $X\setminus K$. Again by the strong maximum principle, $w=0$ or $w>0$ on $K$. If $w=0$ on $K$, then arguing as above, we get that $w=v=0$ on $\partial K$ and also $u\leq 0$ on $K\cup \partial K$. If $w>0$ on $K$ and $w=v=0$ on $\partial K$, then we have uniqueness of the solutions by Proposition~\ref{prop:PP}, i.e., $w=v$ and hence, $u\leq v$ in $K$. If $w>0$ on $K$ and $w=v\gneq 0$ on $\partial K$, then we have uniqueness of the solutions by Lemma~\ref{lem:WCPEinkesselungh>0}, i.e., $w=v$ and hence, $u\leq v$ in $K$.
	
	Secondly, assume that \ref{prop:WCPh>0u} holds. The proof is similar to \ref{prop:WCPh>0v}, but here are the details: Let $g:=Hu$ and $f:=u$ on $X$, and consider for a function $\tilde{u}\in \FF(K)$ the problem 
		\begin{equation}
	\begin{aligned}\label{eq:WCP32}
	\begin{cases}
		H\tilde{u}= g \quad \text{in } K, \\
		\phantom{H}\tilde{u}= f \quad \text{on } X\setminus K.
	\end{cases}
	\end{aligned}
	\end{equation}	
	Then  $C^{-1} u$ is a subsolution of \eqref{eq:WCP32}. By Lemma~\ref{lem:WCPEinkesselungh>0}, there exists a unique solution $w\in C(\KK)$ of \eqref{eq:WCP31} such that $C^{-1}u\leq w\leq v$ on $K$ and $w=u$ on $X\setminus K$. By the strong maximum principle, $w=0$ or $w>0$ on $K$. If $w=0$ on $K$, then arguing as above, we get that $w=u=0$ on $\partial K$ and also $u\leq 0$ on $K\cup \partial K$ which is a contradiction to \ref{prop:WCPh>0u}. If $w>0$ on $K$ and $w=u=0$ on $\partial K$, then we have uniqueness of the solutions by Proposition~\ref{prop:PP}, i.e., $w=u$ and hence, $u\leq v$ in $K$. If $w>0$ on $K$ and $w=u\gneq 0$ on $\partial K$, then we have uniqueness of the solutions by Lemma~\ref{lem:WCPEinkesselungh>0}, i.e., $w=u$ and hence, $u\leq v$ in $K$.
	
	Now, let $K$ be possibly disconnected. Then, we can apply the previous consideration to every connected component of $K$. This yields the result.
\end{proof}

\begin{remark}
We say that the \emph{strong comparison principle} holds true for $h$, if the conditions in Proposition~\ref{prop:WCPh>0} imply $u<v$ on $K$ unless $u=v$ on $K$. 	

In the linear case, i.e., $p=2$, it is shown in \cite[Lemma~5.14]{KePiPo1}, that the strong maximum principle holds true for $h$ on any finite subset. For $p\neq 2$ it is not known if the strong and the weak comparison principle are equivalent (apart from the trivial case $u=0$ where it is a consequence of the Harnack inequality). In the continuum, a very nice discussion is given in \cite[Section~3]{FP11}.

However, if $c\geq 0$ on a connected and finite $K$, then Lemma~\ref{lem:WCPNonNegativePotential}, says that the strong comparison principle holds true for $h$. So further investigations are needed  on not non-negative potentials and $u\neq 0$.

We come back to this notion of strong comparison in Proposition~\ref{prop:SCP} in the context of minimal growth.
\end{remark}

\subsection{Consequences for the Variational Capacity}\label{sec:capa}

As the name suggests, there are many different capacities. A good overview is the monograph \cite{Bjoern}, where many of them are discussed in detail in the quasi-linear Lapacian setting on specific metric spaces which include also metric graphs. For results on local $p$-Schrödinger operators connecting variational capacity and criticality see \cite{PT07, PT08, PT}. Results for the standard $p$-Laplacian on locally finite graphs can be found in \cite{Prado}.

The weak comparison principle allows us to prove the statement that if the capacity vanishes at some vertex, it vanishes at all vertices. Using this result, we can obtain a Green's function globally. But before showing this, we need the following auxiliary lemma.
\begin{lemma}\label{lem:convergenceCap}
	Let $h\geq 0$ on $C_c(V)$ for some $V\sse X$. Let $(V_n)$ be an increasing exhaustion of $V$, and $o\in V_1$. Then, \[\cc_{h}(o,V)=\lim_{n\to\infty}\cc_{h}(o,V_n).\] 
\end{lemma}
\begin{proof}
	By the monotonicity of the infimum we have for all $n\in\NN$ 
	\[\cc_{h}(o,V_{n})\geq \cc_{h}(o,V_{n+1})\geq \cc_{h}(o,V)\geq 0.\]
	Thus, $(\cc_{h}(o,V_n))$ is a decreasing and bounded sequence. Let $\epsilon> 0$ and choose $\phi\in C_c(X)$ such that $h(\phi)\leq \cc_{h}(o,V)+\epsilon$. Then, there is a $n_{0}\in \NN$ such that for all $n\geq n_0$, we have  $\supp \phi \in V_n$, and thus
	\[\cc_{h}(o,V_n)\leq h(\phi)\leq \cc_{h}(o,V)+ \epsilon.\]
	Letting $\epsilon \to 0$ yields the result.
\end{proof}
Now, we can prove the main result of this subsection.

\begin{proposition}\label{prop:cap=0}
	Let $h\geq 0$ on $C_c(V)$ for some connected $V\sse X$. If there is $x\in V$ such that $\cc_{h}(x,V)=0$ then $\cc_{h}(y,V)=0$ for all $y\in V$.
	
	In particular, if $h$ is subcritical in $V$, then $\cc_{h}(y,V)>0$ for all $y\in V$.
\end{proposition}
\begin{proof}
	Let $(K_n)$ be an increasing exhaustion of $V$ with finite and connected sets such that $x\in K_1$. Since $h\geq 0$ on $C_c(V)$, we have using Corollary~\ref{cor:principal} that $\lambda_{0}(K_n)>0$ for all $n\in \NN$. Thus, we can use Lemma~\ref{lem:J} and get the existence of a function $\phi_n\in C(K_n)$ which minimises $h$ on $K_{n,o}:=\set{\phi\in C(K_n): \phi(x)=1}$. By Proposition~\ref{prop:PP}, we get that $\phi_n>0$ on $K_n$. By the weak comparison principle, Proposition~\ref{prop:WCPh>0}, we get that $(\phi_n)$ is increasing. Moreover, using Lemma~\ref{lem:convergenceCap} we get
\[0=\cc_{h}(x,V)=\lim_{n\to\infty}\cc_{h}(x,K_n)=\lim_{n\to\infty}h(\phi_n).\]
	Furthermore, for any $y\in V$, there exists $n_{0}$ such that $y\in K_n$ for all $n\geq n_{0}$. Then, using that $(1/\phi_n (y))_{n\geq n_0}$ is bounded and decreasing,  we compute using Lemma~\ref{lem:convergenceCap}, 
\[ 0\leq\cc_{h}(y,V)=\lim_{n\to\infty}\cc_{h}(y,K_n)\leq \lim_{n\to\infty}\frac{ h(\phi_n)}{\phi^{p}_n(y)}=0.\]
Thus, $\cc_{h}(y,X)=0$.

Let now $h$ be subcritical in $V$, then there exists $o\in V$ such that $w(o)>0$ for some non-negative function $w\in C(V)$, and $h\geq \norm{\cdot}_{p,w}^{p}$ on $C_c(V)$. Thus, $\cc_{h}(o,V)\geq w(o)>0$. Assume that there is $x\in V$ such that $\cc_{h}(x,V)=0$. Then, by the first part, $\cc_{h}(y,V)=0$ for all $y\in V$, which is a contradiction. Hence, $\cc_{h}(x,V)>0$.
\end{proof}

A consequence of the previous proposition is the following statement.

\begin{corollary}\label{cor:w>0}
	Let $V\sse X$ be connected. If $h$ is subcritical in $V$ with corresponding positive Hardy weight $w\in C(V)$, then $w$ can chosen to be strictly positive on $V$.
\end{corollary}
\begin{proof}
	By Proposition~\ref{prop:cap=0}, we have that $\cc_{h}(x,V)> 0$ for all $x\in V$. Moreover, $\cc_{h}(x,V)\cdot 1_x$ is a possible $w$ since $\cc_{h}(x,V)\abs{\phi(x)}^{p}\leq h(\phi)$ for all $\phi\in C_c(V)$. Furthermore, let $\alpha_x>0$ such that $\sum_{x\in V}\alpha_x=1$, then also $\sum_{x\in V}\alpha_x \cc_{h}(x,V)\cdot 1_x$ is a possible $w$ (which is bounded from above pointwise by the  first possible weight and thus the sum is convergent), and thus $w$ can chosen to be strictly positive on $V$.
\end{proof}
 
 By Corollary~\ref{cor:w>0}, we know that any subcritical energy functional has a strictly positive Hardy weight. If we know a little bit more about the lower bound, we get a connection to the principal eigenvalue. This is specified next. In the case of finite $V\sse X$, it gives another characterisation of the maximum principle and continues Proposition~\ref{prop:PP}.
 
\begin{proposition}\label{prop:subLam}
 	Let $V\sse X$. Then the following holds:
		\begin{enumerate}[label=(\alph*)]
			\item\label{prop:subLam1} If $\lambda_{0}(V)>0$, then $h$ is subcritical in $V$ with strictly positive Hardy weight $w=\lambda_{0}(V)\cdot m$. 
			\item\label{prop:subLam2} If $h$ is subcritical in $V$ with Hardy weight $w$ such that $\inf_V(w/m)> 0$, then $\lambda_{0}(V)>0$. Moreover, if $K\sse X$ is finite and $h$ is non-negative in $C_c(K)$, then $\lambda_{0}(K)>0$.
		\end{enumerate}		 	
 	In particular, if $h$ is non-negative in $C_c(X)$, then $ \lambda_{0}(K)>0$ in every finite  subset $K\sse X$.
 \end{proposition} 
 \begin{proof}
 	Ad~\ref{prop:subLam1}: If $\lambda_{0}(V)>0$, then we have for any $\phi\in C_c(V)$ that $h(\phi)\geq \lambda_0(V)\norm{\phi}^{p}_{p,m}$. Defining $w=\lambda_{0}(V)\cdot m$, we have a possible strictly positive Hardy weight, and $h$ is subcritical in $V$.
 	
 	Ad~\ref{prop:subLam2}: Since $\inf_{V}(w/m) > 0$, we have
 	\[\lambda_{0}(V)=\inf_{\phi\in C_c(X)\setminus\set{0}}\frac{h(\phi)}{\norm{\phi}^{p}_{p,m}}\geq \inf_{\phi\in C_c(X)\setminus\set{0}}\frac{\inf_{V}(w/m)\norm{\phi}^{p}_{p,m}}{\norm{\phi}^{p}_{p,m}}=\inf_{V}(w/m)>0.\]
 	
 	The second statement can be seen as follows: Indeed, if $K$ is finite and $h$ is non-negative in $C_c(K)$, then by Corollary~\ref{cor:AP}, $h$ is subcritical in $K$. Thus, by Corollary~\ref{cor:w>0}, there is a strictly positive Hardy weight $w$ on $K$. Since  $\inf_K(w/m)=\min_K (w/m) > 0$, we can apply the first statement in  \ref{prop:subLam2}, and get the desired assertion.
 	
 	The last statement follows also from Corollary~\ref{cor:AP} because if $h$ is non-negative in $C_c(X)$ it is also non-negative in $C_c(K)$.
 	 \end{proof}

 	 \begin{remark} By Corollary~\ref{cor:AP}, the statement in Proposition~\ref{prop:subLam} is more about the choice of a strictly positive Hardy weight than about subcriticality.
 	 \end{remark}

\section{Global Results for (Sub-)Critical Energy Functionals}  \label{sec:criticalEnergyFunctionals}
With our developed toolbox from the previous sections it is possible to give some characterisations of (sub-)criticality. Recall that an energy functional is called subcritical if the Hardy inequality holds. Recently, a list of characterisations connecting (sub-)criticality with the (non-)existence of so-called null-sequences and with harmonic functions was shown in \cite{F:GSR}. Here, we continue this list with the connection to the (non-)existence of Green's functions.

\subsection{Existence and Properties of Global Green's Functions}\label{sec:Green}

Recall that we want to show that $h$ is subcritical in a connected set $V\sse X$ if and only if a normalised minimal positive Green's function exists in $V$. By Corollary~\ref{cor:AP}, $h$ is subcritical in every proper subset of $X$ if and only if it is non-negative on $C_c(X)$. Thus, on proper and connected subsets, we show that we always have a Green's function.

\begin{proof}[Proof of Theorem~\ref{thm:GreensFunction}]
Ad "$\implies$": Since $h$ is subcritical in $V$, there exists a positive function $w\in C(V)$ such that $h\geq \norm{ \cdot}^{p}_{p,wm}$, i.e., $w\cdot m$ is a Hardy weight. Hence, there exists $o\in V$ such that $w(o)> 0$. Let $(K_n)$ be an increasing exhaustion of $V$ with finite and connected sets, and $o\in K_1$. Since $h$ is non-negative on $C_c(V)$, we get from Corollary~\ref{cor:principal} that $\lambda_{0}(K_n)>0$ for any $K_n$. By the lemma about the solutions of Poisson-Dirichlet problems, Lemma~\ref{lem:J}, we get the existence of a positive function $u_n\in C(K_n)$ which is harmonic on $K_n\setminus \set{o}$, $u_n(o)=1$, and which minimises $h$ on $K_{n,o}:=\set{\phi\in C(K_n): \phi(o)=1}$. 

Note that for all $t\in \RR$, we have $(1-t)u_n+t 1_o\in K_{n,o}$. By the definition of being a minimiser, the function $t\mapsto h((1-t)u_n+t 1_o)$ has derivative zero at $t=0$. Thus,
\[0=\left.\frac{\dd }{\dd t}h((1-t)u_n+t 1_o)\right|_{t=0}=-p h(u_n)+ p Hu_n(o)m(o).\]
Rearranging and using that  $h$ is subcritical in $V$ yields in 
\[Hu_n(o)m(o)= h(u_n)\geq w(o)m(o)>0.\]
Hence, $u_n$ is even strictly superharmonic on $\set{o}$, and in particular superharmonic on $K_n$. Because of $u_n(o)=1$, we get by the Harnack inequality that $u_n$ is strictly positive on $K_n$. By the characterisations of the maximum principle on finite subsets, Proposition~\ref{prop:PP}~'\ref{prop:PP0}$\implies$\ref{prop:PP3}', we have the existence of a unique positive solution $v_n\in C(K_n)$ such that $Hv_n=C_n\cdot 1_o$, where the constant is given by
\[C_n:=\cc_{h}(o,K_n)/m(o)\geq 0.\] 
Hence, $u_n=v_n$ and $u_n$ uniquely minimises $j_{C_n\cdot 1_o}$ on $C(K_n)$. Clearly, $(C_n)$ is a decreasing sequence. 
Since obviously \[\cc_{h}(o,V)=\inf_{\phi\in C_c(V), \phi(o)=1}h(\phi)\geq \inf_{\phi\in C_c(V), \phi(o)=1}\norm{\phi}_{p,wm}^{p}\geq w(o)m(o)>0, \]
we get
\[ C_n\geq \frac{\cc_{h}(o,V)}{m(o)}\geq w(o)> 0.\]

Furthermore, note that 
\begin{align*}
	\begin{cases}
		Hu_n=Hu_{n+1}=0 &\qquad \text{on } K_{n}\setminus \set{o}, \\
		\phantom{H}u_{n}\leq u_{n+1} &\qquad \text{on } (K_{n+1}\setminus K_{n})\cup\set{o}.
	\end{cases}
\end{align*}
Since $\lambda_{0}(K_n\setminus\set{o})\geq \lambda_{0}(K_n)>0$, we can apply the weak comparison principle, Proposition~\ref{prop:WCPh>0}, and get that $(u_n)$ is increasing pointwise on $V$.

Since for all $n\in \NN$, we have $u_n(o)=1$, we can apply the Harnack principle, Lemma~\ref{lem:compactBC}, and get that the pointwise limit $u$ exists, and  $Hu_n \to Hu$ on $X$ pointwise as $n\to\infty$.

Another application of the weak comparison principle shows that $u$ is independent of the choice of the sequence $(K_n)$, i.e., $u$ is uniquely determined.

Since $Hu_n=C_n1_o$ on $K_n$, we infer using Lemma~\ref{lem:convergenceCap} that $Hu=\lim_{n\to\infty}C_n1_o= \cc_{h}(o,V)1_o/m(o)$ on $V$.  By Proposition~\ref{prop:cap=0}, we get from $\cc_{h}(o,V)>0$, that $\cc_{h}(x,V)>0$ for all $x\in V$. Thus, we can do this construction for all $x\in V$.

We define for every $o\in V$ the function $G\colon V \to (0,\infty)$ via
\[G_o(y)=\left(\frac{m(o)}{ \cc_{h}(o,V)}\right)^{\frac{1}{p-1}}u(y),\]
and have a function which satisfies $HG_o=1_o$ and thus, a candidate for the desired normalised Green's function. 

We show now that $G_o\in \MM(V\setminus \set{o})$: Let $o\in K\sse V$, where $K$ is finite and connected. Let $0\leq v\in \FF(V\setminus K)$ be superharmonic on $V\setminus K$ and $v(o)\geq u(o)=1$. Since $u$ is independent of the choice of the exhaustion, we can assume that there is an $n\in \NN$ such that $K_n=K$. Then by the weak comparison principle, Proposition~\ref{prop:WCPh>0}, $v\geq u_n$ on every connected component of $K_n\setminus \set{o}$, and thus, on $K_n$ for every $n\in \NN$. Hence, $v\geq u$ on $X$. Using the $(p-1)$-homogeneity, we get that $G_o\in \MM(V\setminus \set{o})$ is a normalised Green's function.

%
%
%

Ad "$\Longleftarrow$": By assumption, we have a strictly positive superharmonic function in $V$, which is not harmonic in $V$. Now, Lemma~\ref{lem:groundstate} implies that $h$ is not critical.
\end{proof}
As a direct consequence of the previous proof of Theorem~\ref{thm:GreensFunction}, we get the following properties of Green's functions.
\begin{corollary}[Properties of Green's Functions]\label{cor:GreensFunction}
	Let $V\sse X$ be connected. Assume that $h$ is subcritical in $V$, then a normalised Green's function $G_o$ at $o\in V$ is unique, in $\DD$, and for all $o\in V$ we have
	\[h(G_o)=m^{p/(p-1)}(o)\cc^{-1/(p-1)}_{h} (o,V).\]
Furthermore,  if $c(o)=0$, then $G_o$ is not constant.
	\end{corollary}	
	\begin{proof}
Recall the construction	 of the Green's function in the proof of Theorem~\ref{thm:GreensFunction}. It is easy to see that $G_o\in \DD$, that it is unique and that it fulfils the desired equality.
	
		The last statement can be seen as follows: It is obvious, that a function $f$ which is constant for all $x\sim o\in V$, is $L$-harmonic in $\set{o}$. Since $LG_o(o)=HG_o(o)=1$, we conclude that $G_o$ is not constant.	
	\end{proof}

\subsection{Minimal Growth}\label{sec:growth}
Let $V\sse X$ be connected, and $K\sse V$ be finite. Recall the definition of a harmonic function in $V\setminus K$ of minimal growth in a neighbourhood of infinity in $V$, Definition~\ref{def:minimal}: $u\in \MM(V\setminus K)$ if and only if  $u$ is harmonic on $V\setminus K$,  positive on $V\cup \partial V$, and for any  finite and connected subset $\KK \sse V$ with $K\sse \KK$, and any positive function $v\in \FF(V\setminus \KK)$ which is superharmonic in $V\setminus \KK$, we have
	\[u\leq v \text{ on } \KK \quad\text{ implies }\quad  u\leq v \text{ in } V\setminus \KK.\]
We will discuss some properties of such functions here.



Note that $\MM(V\setminus K)\sse \MM(V\setminus \KK)$ for all finite $K\sse \KK\sse V$. On the other hand, the inverse assertion seems to depend on the strong comparison principle, confer Subsection~\ref{sec:weak}. To show this, is the goal of this subsection.

Without loss of generality we only need to consider the case $K\subsetneq V$. Let $h$ be non-negative on $C_c(V)$, and let $v$ be a positive superharmonic function in $V$ which exists by the Agmon-Allegretto-Piepenbrink theorem, Theorem~\ref{thm:AP}. Corollary~\ref{cor:principal} implies that $\lambda_{0}(\KK)>0$ for any finite and connected $\KK\sse V$. Let $(K_n)$ be an increasing exhaustion of $V$ with finite and connected sets such that $K\subsetneq K_0$. Let $u\in C(K)$ be an arbitrary positive function, i.e., $\supp (u)\sse K$ and $u>0$ on $K$. By Lemma~\ref{lem:J} there exists a positive solution $u_n\in C(K_n) $ of the following Dirichlet problem
\begin{align*}
	\begin{cases}
	Hw=0 &\text{ in } K_n\setminus K,\\
	\phantom{H}w=u &\text{ in } K.
	\end{cases}
\end{align*}
By the weak comparison principle, Proposition~\ref{prop:WCPh>0}, and the strong maximum principle, Proposition~\ref{prop:PP}, $(u_n)$ is a monotone increasing sequence.  Set $C:= \max_{x\in K}(u(x)/v(x))>0$, then again by the weak comparison principle, $u_n\leq C\,v$ on $V$. Define the pointwise limit 
\[u^{K}:=\lim_{n\to\infty}u_n\geq 0.\] 
Applying the weak comparison principle once more, we see that $u^{K}$ does not depend on the choice of the exhaustion. By the convergence of solution principle, Proposition~\ref{prop:harnack}, we get $Hu^{K}=0$ on $V\setminus K$. 

If $v\geq u$ on $K$, then by the weak comparison principle, Proposition~\ref{prop:WCPh>0}, $v\geq u_n$ on any $K_n$, $n\in \NN$. Thus, $v\geq u^{K}$, and $u^{K}\in \MM(V\setminus K)$.

Similarly, one shows the following lemma, confer \cite[Lemma~9.4]{PR15}.
\begin{lemma}\label{lem:SCP}
	Let $h\geq 0$ on $C_c(V)$ where $V\sse X$ is connected. Let $K$ be finite. Let $u$ be a positive function that is harmonic on $V\setminus K$. Then $u\in\MM(V\setminus K)$ if and only if $u=u^{\KK}$ for any finite and connected $K\sse \KK$. 
\end{lemma}


The following result is the discrete analogue to \cite[Proposition~5.2]{PT08}. Recall the definition of the strong comparison principle from the remark in Section~\ref{sec:weak}: $h$ fulfils strong comparison principle on a finite and connected $K\sse X$ if the conditions in Proposition~\ref{prop:WCPh>0} imply $u<v$ on $K$ unless $u=v$ on $K$. 	
\begin{proposition}\label{prop:SCP}
	Let $h\geq 0$ on $C_c(V)$ where $V\sse X$ is connected. Assume that the strong comparison principle holds true for $h$ on any finite and connected subset. Let $K\sse \KK \sse V$ be two finite sets. Assume that there exists a positive function $u\in \MM(V\setminus \KK)$ which is harmonic in $V\setminus K$. Then $u\in \MM(V\setminus  K)$.
\end{proposition}
\begin{proof}
	The case $K=\KK$ is evident. Thus, assume that $K\subsetneq \KK\sse V$. Moreover, let $W \sse \WW\sse V$ be two finite and connected sets such that $K\sse W$ and $\KK\sse \WW$. Since $u\in \MM(V\setminus\KK)$, we have using Lemma~\ref{lem:SCP}, that $u=u^{\WW}$. Hence, in particular, $u=u^{\WW}$ on  $V\setminus \WW$. Since $W$ is finite, we also have $u\asymp u^{W}$ in $W$. Using the weak comparison principle, Proposition~\ref{prop:WCPh>0}, and an exhaustion argument, it follows that $u\asymp u^{W}$ in $X\setminus W$. 
	
	Set \[\epsilon^{W}:=\max\set{\epsilon > 0: \epsilon u \leq u^{W} \text{ in } X\setminus W}.\]
	Then, since $K\sse \KK$ and $u\in \MM(V\setminus\KK)$, we have $0< \epsilon^{W}\leq 1$. Assume that we do not have  equality, i.e., assume that $\epsilon^{W}<1$. Then, since $\epsilon^{W}u\lneq u^{W}$ in $X\setminus W$ and $\epsilon^{W}u < u^{W}$ in $W$, we get by the strong comparison principle that $\epsilon^{W}u< u^{W}$ in $W_n\setminus W$ for every increasing exhaustion $(W_n)$ of $V$ with finite and connected sets. Therefore, there exists $ \tilde{\epsilon}> 0$ such that $(1+\tilde{\epsilon})\epsilon^{W}<1$ and $(1+\tilde{\epsilon})\epsilon^{W}u \leq u^{W}$ on $\WW \setminus W$, and thus on $X\setminus \WW$. Hence, $(1+\tilde{\epsilon})\epsilon^{W}u \leq u^{W}$ on $X\setminus W$, but this a contradiction to the definition of $\epsilon^{W}$. Thus, $u=u^{W}$ in $X\setminus W$, and therefore, $u\in \MM(V\setminus K)$. 
\end{proof}

\subsection{Proof of the Characterisations of Criticality}\label{sec:critical}

\begin{proof}[Proof of Theorem~\ref{thm:critical}]
	Ad~\ref{thm:critical1} $\implies$ \ref{thm:critical9}: If $h$ is critical in $X$, then by Lemma~\ref{lem:groundstate}, every positive superharmonic function in $X$ is harmonic in $X$. Hence, there cannot exists a Green's function in $X$.
	
	Ad~\ref{thm:critical9} $\implies$ \ref{thm:critical1}: This is Theorem~\ref{thm:GreensFunction}.

	Ad~\ref{thm:critical11} $\implies$ \ref{thm:critical9}: Assume that $u>0$ is a minimal positive harmonic function in $X$. Furthermore, assume that there exists a Green's function in $X$, i.e., a positive function $v\in \FF$ such that $Hv\gneq 0$ on $X$.
	
	Let $K\sse X$ be finite and connected, and set $\epsilon=\max_{x\in K}\set{u(x)/v(x)}$. Since $u$ is a minimal positive harmonic function in $X$, and $\epsilon u \leq v$ on $K$, we have $\epsilon u \leq v$ on $V\setminus K$. Moreover, we have $\epsilon u \neq v$, since otherwise $v$ would be harmonic, which is a contradiction. Thus, there exists a finite and connected subset $\KK$ of $X$  and $\tilde{\epsilon}>0$ such that $(1+\tilde{\epsilon})\epsilon u \leq v$ on $\KK$. But since $u$ is a global minimal positive harmonic function, we get $(1+\tilde{\epsilon})\epsilon u \leq v$ on $X$, and in particular, on $K$, which is a contradiction to the maximality of $\epsilon$ on $K$. Hence, we cannot have a Green's function on $X$.
	
	Ad~\ref{thm:critical1} $\implies$ \ref{thm:critical11}: By \cite[Theorem~5.1]{F:GSR}, the criticality of $h$ implies that we have a unique superharmonic function $u$ (up the linear dependence), and this function is harmonic. This is the so-called \emph{Agmon ground state}. Hence, there is nothing to prove for $K=\emptyset$.
	
Take an arbitrary exhaustion $(K_n)$ of $X$ with finite and connected sets, and some $o\in K_1$. Without loss of generality, we can assume that $u(o)=1$. Since $h\geq 0$, we get $\lambda_{0}(K_n)>0$ for every $n\in \NN$ by Corollary~\ref{cor:principal}. By Proposition~\ref{prop:PP}, we get the existence of a sequence $(v_n)$ in $C(K_n)$ such that $Hv_n= g_n$ on $K_n$ for any $0\leq g\in C(K_n)$. By the weak comparison principle, Proposition~\ref{prop:WCPh>0}, we get that $(v_n)$ is monotone increasing. Set $w_n:=v_n/v_n(o)$.

If $(v_n(o))$ is bounded, then by the Harnack principle, Lemma~\ref{lem:compactBC}, we get that $v_n$ converges pointwise to a function $v$, and $Hv_n \to Hv$, i.e., $Hv=g$. Since $h$ is critical, we get a contradiction unless $g=0$, and $v/v(o)= u$. 

If $(v_n(o))\to \infty$ as $n\to\infty$, we consider instead $w_n$. Then $(w_n)_{n\geq N}$ is in $S^{+}_{o}(K_N)$ for any $N\in \NN$. Thus, by the Harnack principle, Proposition~\ref{prop:harnack}, there exists a subsequence of $(w_n)$ that converges to a positive and superharmonic function $w^{N}$ on $K_N$. Letting $N\to\infty$, we see that the limit $w$ is positive and superharmonic on $X$ which is again a contradiction unless $w=u$. 

Now, let $K$ be finite, and $\tilde{v}\in \FF(X\setminus K)$ be a positive superharmonic function on $X\setminus K$ and $u\leq \tilde{v}$ on $K$. Then for any $\epsilon >0$ there exists $n_{\epsilon}$ such that for all $n\geq n_{\epsilon}$, we have $H\tilde{v}\geq Hw_n=0$ on $K_n\setminus K$, and $0\leq w_n\leq (1+\epsilon)\tilde{v}$ on $K\cup X\setminus K_n$. On any connected component of $K_n\setminus K$, we get by the weak comparison principle that $w_n\leq (1+\epsilon)\tilde{v}$ on $X\setminus K$ for any $n\in\NN$. Thus, $u\leq (1+\epsilon)\tilde{v}$ on $X\setminus K$. Letting $\epsilon \to 0$ we obtain $u\leq \tilde{v}$ on $X\setminus K$.

Ad~\ref{thm:critical1} $\iff$ \ref{thm:critical3}: This follows from Proposition~\ref{prop:cap=0} and \cite[Theorem~5.1]{F:GSR}.
\end{proof}

\begin{remark}
\begin{enumerate}[label=(\alph*)]
\item Let $0\lneq \phi\in C_c(X)$. A function  $G_{\phi}\in \MM(X\setminus \supp (\phi))$ such that $HG_{\phi}=\phi$ is called \emph{Green potential} with \emph{charge} $\phi$.

By a mild generalisation of the results in Subsection~\ref{sec:Green} together with the approach in Subsection~\ref{sec:growth}, it is not difficult to see that $h$ is subcritical if and only if there exists a Green potential for some (any) charge $0\lneq \phi\in C_c(X)$. Confer also \cite[Lemma~3.2]{Versano} for the analogue result in the continuum.

\item A slightly more general way of defining capacity is by setting $\cc_h(K,V):=\inf\set{h(\phi) : \phi\in C_c(V), \phi \geq 1 \text{ on } K}$ for all finite $K\sse V$, where $V\sse X$ is connected. Using a mild modification of Subsection~\ref{sec:capa}, it is not difficult to show that $h$ is subcritical in $V$ if and only if $\cc_h(K,V)> 0$ for some (all) finite $K\sse V$.
\end{enumerate}
\end{remark}

\subsection{Bounds for the Principal Eigenvalue}
We close this paper by showing upper and lower bounds for $\lambda_0(V), V\sse X$. These results go under the name Barta's inequality, see \cite{Barta} for the original paper by Barta, \cite[Theorem~7.1]{AmghibechPicone} for the corresponding version for standard $p$-Laplacians on finite graphs or \cite[Section~2.2]{AllegrettoHuangPicone} for $p$-Laplacians on the Euclidean space.  
A linear version of Barta's theorem for finite graphs can be found in \cite[Theorem~2.1]{Urakawa}.  For the linear case in the continuum see e.g. \cite{NP92} and references therein.


\begin{proposition}[Barta-type Inequality]\label{prop:Barta}
	Let $V\sse X$.
\begin{enumerate}[label=(\alph*)]
	\item\label{prop:bartainf} If $h\geq 0$ on $C_c(V)$ and $u\in \FF(V)$ such that $u>0$ on $V$. Then,
	\[\inf_{x\in V}\frac{Hu(x)}{u^{p-1}(x)}\leq \lambda_0(V).\]
	\item\label{prop:bartasup} If $u\in\DD$ such that $u> 0$ on $V$ and $u=0$ on $X\setminus V$ then
	\[\lambda_0(V)\leq \sup_{x\in V}\frac{Hu(x)}{u^{p-1}(x)}.\] 
\end{enumerate}	
	 \end{proposition}
\begin{proof}
	Ad~\ref{prop:bartainf}: Let $I=\inf_{x\in V}Hu(x)/u^{p-1}(x)$. Since $h\geq 0$ on $C_c(V)$, we have $\lambda_0(V)\geq 0$. Hence, the case $I\leq 0$ is trivial. If $I>0$, then we can use Picone's inequality or Lemma~\ref{lem:groundstate} with $g=I$ to get
	\[ h(\phi)\geq I\cdot \norm{\phi}^p_{p,m},\qquad \phi\in C_c(V).\]
	This implies $\lambda_0\geq I$.
	
	Ad~\ref{prop:bartasup}: Notice that $\DD\sse \FF(V)$ by Green's formula, Lemma~\ref{lem:GreensFormula}. Let $S=\sup_{x\in V}Hu(x)/u^{p-1}(x)$. Then, $Hu\leq S\,u^{p-1} $ on $V$. Hence,
	\[\sum_{x\in V}Hu(x)u(x)m(x)\leq S\, \sum_{x\in V}u^{p}(x)m(x),\]
	where both sides might take the value $+\infty$. Using $u=0$ on $X\setminus V$ and again Green's formula, Lemma~\ref{lem:GreensFormula}, we have $h(u)\leq S\, \norm{u}^p_{p,m}$, whenever $u\in \DD$. Hence,  we get $\lambda_0(V)\leq S$. 
\end{proof}
Barta's inequality has the following simple consequence.
\begin{corollary}
Let $h$ be non-negative on $C_c(X)$, then we have
\[0\leq \lambda_0(X)\leq  \inf_{x\in X} \frac{\cc_h(x,X)}{m(x)}. \]
\end{corollary}
\begin{proof}
If $h$ is critical in $X$, then clearly $\lambda_0(X)=0=\cc_h(x,X)$.

Let $h$ be subcritical in $X$. Firstly, we note that
\[ \lambda_0(X)\leq \inf_{x\in X} 1/(G_x(x))^{p-1,}\]
where $G_o$ is the Green's function normalised at $o\in X$. Indeed, by Theorem~\ref{thm:GreensFunction}, we have that $0<G_o\in \DD$ for all $o\in X$. Hence, we can apply Proposition~\ref{prop:Barta}, which yields the inequality. 
	
	Now, the statement of the corollary follows from the construction of the Green's function in the proof of Theorem~\ref{thm:GreensFunction}.
\end{proof}


\noindent\textbf{Acknowledgements.} The author wishes to express his sincere gratitude to Matthias Keller who initiated the main problem studied in this paper, and for the support  he gave him. Furthermore, the author thanks the Heinrich-Böll-Stiftung for the support.

\printbibliography
\end{document}